\documentclass{article}

\usepackage{arxiv}

\usepackage[utf8]{inputenc} 
\usepackage[T1]{fontenc}    
\usepackage{hyperref}       
\usepackage{url}            
\usepackage{booktabs}       
\usepackage{amsfonts}       
\usepackage{nicefrac}       
\usepackage{microtype}      
\usepackage{lipsum}

\usepackage{cite}
\usepackage{amsmath,amssymb,amsfonts}
\newenvironment{pf}{\begin{proof}}{\end{proof}}
\newenvironment{assum}{\begin{assumption}}{\end{assumption}}
\newenvironment{defn}{\begin{definition}}{\end{definition}}
\newenvironment{prob}{\begin{problem}}{\end{problem}}
\newenvironment{thm}{\begin{theorem}}{\end{theorem}}
\usepackage{graphicx}
\usepackage{algorithm}
\usepackage[noend]{algpseudocode}
\usepackage{hyperref}
\usepackage{textcomp}
\usepackage{threeparttablex} 
\usepackage{adjustbox}
\usepackage{mathtools}
\usepackage{caption}
\usepackage{float}
\usepackage{stmaryrd}
\usepackage{epstopdf}
\usepackage{multirow}
\usepackage{pifont}
\usepackage{caption}
\usepackage{subcaption}
\usepackage{xcolor}

\newcommand{\norm}[1]{\left\lVert#1\right\rVert}

\newcommand{\R}{{\mathbb{R}}}

\newcommand{\B}{{\mathcal B}}

\newcommand{\N}{{\mathbb{N}}}

\newcommand{\T}{{\mathbf{T}}}

\newcommand{\So}{{\mathbf{S}}}
\newcommand{\Obs}{{\mathcal{U}}}

\newcommand{\U}{{\mathbf{U}}}

\newcommand{\cen}{\mathbf{c}}
\newcommand{\rad}{\mathbf{r}}
\newcommand{\soc}{{s}_a}
\newcommand{\Ag}{{\mathcal{A}}}
\newcommand{\card}[1]{{\text{card}(#1)}}

\newcommand{\cmark}{\ding{51}}%
\newcommand{\xmark}{\ding{55}}%

\newtheorem{theorem}{Theorem}[section]

\newtheorem{assumption}{Assumption}

\newtheorem{definition}[theorem]{Definition}
\newtheorem{lemma}[theorem]{Lemma}
\newtheorem{remark}[theorem]{Remark}
\newtheorem{problem}[theorem]{Problem}
\newtheorem{proof}[theorem]{Proof}

\title{Incorporating Social Awareness into Control of Unknown Multi-Agent Systems: A Real-Time Spatiotemporal Tubes Approach
\thanks{ S. Upadhyay, R. Das, and P. Jagtap  are with the Robert Bosch Centre for Cyber-Physical Systems, IISc, Bangalore, India {\tt\footnotesize \{siddharthau,ratnangshud,pushpak\}@iisc.ac.in}}
}

\author{
 Siddhartha Upadhyay $^\dag$ \\
  Robert Bosch Centre for Cyber-Physical Systems\\
  IISc, Bengaluru, India\\
  \texttt{siddharthau@iisc.ac.in} \\
   \And
 Ratnangshu Das \thanks{Authors contributed equally.} \\
  Robert Bosch Centre for Cyber-Physical Systems\\
  IISc, Bengaluru, India\\
  \texttt{ratnangshud@iisc.ac.in} \\
   \And
 Pushpak Jagtap \\
  Robert Bosch Centre for Cyber-Physical Systems\\
  IISc, Bengaluru, India\\
  \texttt{pushpak@iisc.ac.in} \\
}

\begin{document}
\maketitle
\setcounter{footnote}{0}
\begin{abstract}                          
This paper presents a decentralized control framework that incorporates social awareness into multi-agent systems with unknown dynamics to achieve prescribed-time reach-avoid-stay tasks in dynamic environments. Each agent is assigned a social awareness index that quantifies its level of cooperation or self-interest, allowing heterogeneous social behaviors within the system. Building on the spatiotemporal tube (STT) framework, we propose a real-time STT framework that synthesizes tubes online for each agent while capturing its social interactions with others. A closed-form, approximation-free control law is derived to ensure that each agent remains within its evolving STT, thereby avoiding dynamic obstacles while also preventing inter-agent collisions in a socially aware manner, and reaching the target within a prescribed time. The proposed approach provides formal guarantees on safety and timing, and is computationally lightweight, model-free, and robust to unknown disturbances. The effectiveness and scalability of the framework are validated through simulation and hardware experiments on a 2D omnidirectional
\end{abstract}

\section{Introduction}
Multi-agent systems have attracted considerable research interest in recent years due to their ability to handle complex tasks through cooperative interactions among agents. They had been widely used in areas such as search and rescue \cite{Search}, safety-critical human-robot interaction in healthcare and medical assistance \cite{chakraborty2014medical}, autonomous driving \cite{auto_driving}, cooperative exploration \cite{juvvi2024safe}, and many other domains \cite{dorri2018multi}.

One of the most important aspects in multi-agent systems is ensuring safe interactions among agents, particularly avoiding inter-agent collisions. In the literature, this challenge is often addressed through safe control strategies that assume fully cooperative and symmetric interactions among agents. For example, \cite{van2008reciprocal} introduces the concept of reciprocal velocity obstacles for real-time multi-agent navigation. Barrier certificate–based methods have also been widely explored to ensure safe multi-robot coordination. For example, in \cite{multi_barrier}, safety barrier certificates are used to guarantee collision-free operation in multi-robot systems. Similarly, \cite{multi_drone} extends this concept to teams of differentially flat quadrotors, enabling collision-free maneuvers. \cite{multi_prob} introduces Probabilistic Safety Barrier Certificates (PrSBC), which leverage control barrier functions to handle the challenges of collision avoidance under uncertainty. Although these methods provide provable safety guarantees, they require solving optimization problems that become computationally expensive as system dimensionality and the number of agents increase. Moreover, they rely on known system dynamics, which is often unavailable for real-world systems. The spatiotemporal tube (STT) framework \cite{das2024prescribed}, \cite{das2024spatiotemporal} has emerged as an effective approach for solving complex task specifications \cite{das2025approximation} without requiring exact knowledge of system dynamics. Building on this idea, \cite{multi_stt} introduces a multi-agent negotiation framework that ensures collision-free motion while satisfying prescribed-time reach-avoid-stay (PT-RAS) tasks. Subsequently, \cite{TCNS} extends this framework to handle general time-varying unsafe sets and achieves smoother control performance. However, these existing STT-based frameworks for multi-agent systems synthesize the tubes offline, either through a negotiation-based procedure \cite{multi_stt} or by solving computationally intensive optimization problems \cite{TCNS}. Moreover, they require \emph{a priori} knowledge of obstacle positions and motion. As a result, these approaches are not suitable for dynamic environments that require real-time replanning.

However, these approaches neglect an important aspect of real-world multi-agent interactions: agents may differ in their cooperation levels. Assuming symmetric interactions is often unrealistic as agents have different social awareness and personality \cite{schwarting2019social}. For instance, a robot delivering groceries is expected to have higher social awareness as compared to an agent urgently delivering medicine. Similarly, as shown in Figure \ref{fig:soc}, at a road intersection, a fire-truck may move straight ahead, while a commercial vehicle should adjust its behavior to prevent collisions. Thus, social awareness can be viewed as a personality trait, and modeling heterogeneous interactions among agents leads to a more realistic representation of such systems.
One common way to quantify social preference or personality in a multi-agent framework is through the concept of Social Value Orientation (SVO) \cite{murphy2011measuring}. It is an idea originating from sociology and psychology, which measures the degree of an agent’s selfishness or altruism. 

In recent years, several works have proposed integrating the concept of SVO into control frameworks, with the aim of modeling heterogeneous interactions among agents. In \cite{pierson2020weighted}, the authors introduce a Weighted Buffered Voronoi tessellation for semi-cooperative multi-agent navigation policies with guarantees on collision avoidance. However, this approach is limited to single-integrator dynamics, and constructing Voronoi cells becomes computationally expensive for higher-dimensional systems. In \cite{he2024priority}, the authors use a similar approach to assign priority levels to the agents to resolve deadlocks. Control barrier function (CBF)-based methods have also been used to solve this problem. For example, \cite{lyu2022responsibility} computes the relative personality differences between agents from a global personality score and designs a CBF-based safe control law. In \cite{risk_aware}, a risk-aware decentralized framework is proposed, where each agent is assigned a proportion of responsibility for collision avoidance, ensuring safe and efficient motion without direct communication. The work in \cite{chandra2023socialmapf} tackles the multi-agent path-finding problem in socially aware settings using an auction-based mechanism to resolve conflicts among agents. But again, these methods depend on state-space discretization or require solving optimization problems at each step, which limits their scalability to higher-dimensional systems and relies on exact knowledge of the system dynamics.


In this work, we explore the integration of social preferences or personality information into multi-agent systems operating under unknown dynamics. The main contributions of this paper are summarized as follows:
\begin{itemize}
    \item We propose a \textit{real-time spatiotemporal tube (STT)} framework for decentralized multi-agent systems with unknown dynamics to solve temporal reach-avoid-stay (TRAS) tasks in dynamic environments.    
    \item We introduce a \textit{social awareness index} to model heterogeneous social behaviors among agents and incorporate it into the STT design to capture socially-aware inter-agent interactions.    
    \item We derive an \textit{approximation-free}, \textit{closed-form control law} that constrains each agent within its evolving STT, thereby ensuring collision avoidance, obstacle avoidance, and prescribed-time target reachability.    
    \item The proposed framework provides formal safety and timing guarantees without requiring offline tube synthesis or \textit{a priori} knowledge of obstacle trajectories, while remaining computationally lightweight and suitable for real-time implementation.
    \item We validate the effectiveness of the proposed approach through extensive simulation and hardware experiments involving 2D omnidirectional mobile robots and 3D aerial vehicles.
\end{itemize}

\begin{figure}[t]
    \centering
    \includegraphics[width=0.42\textwidth]{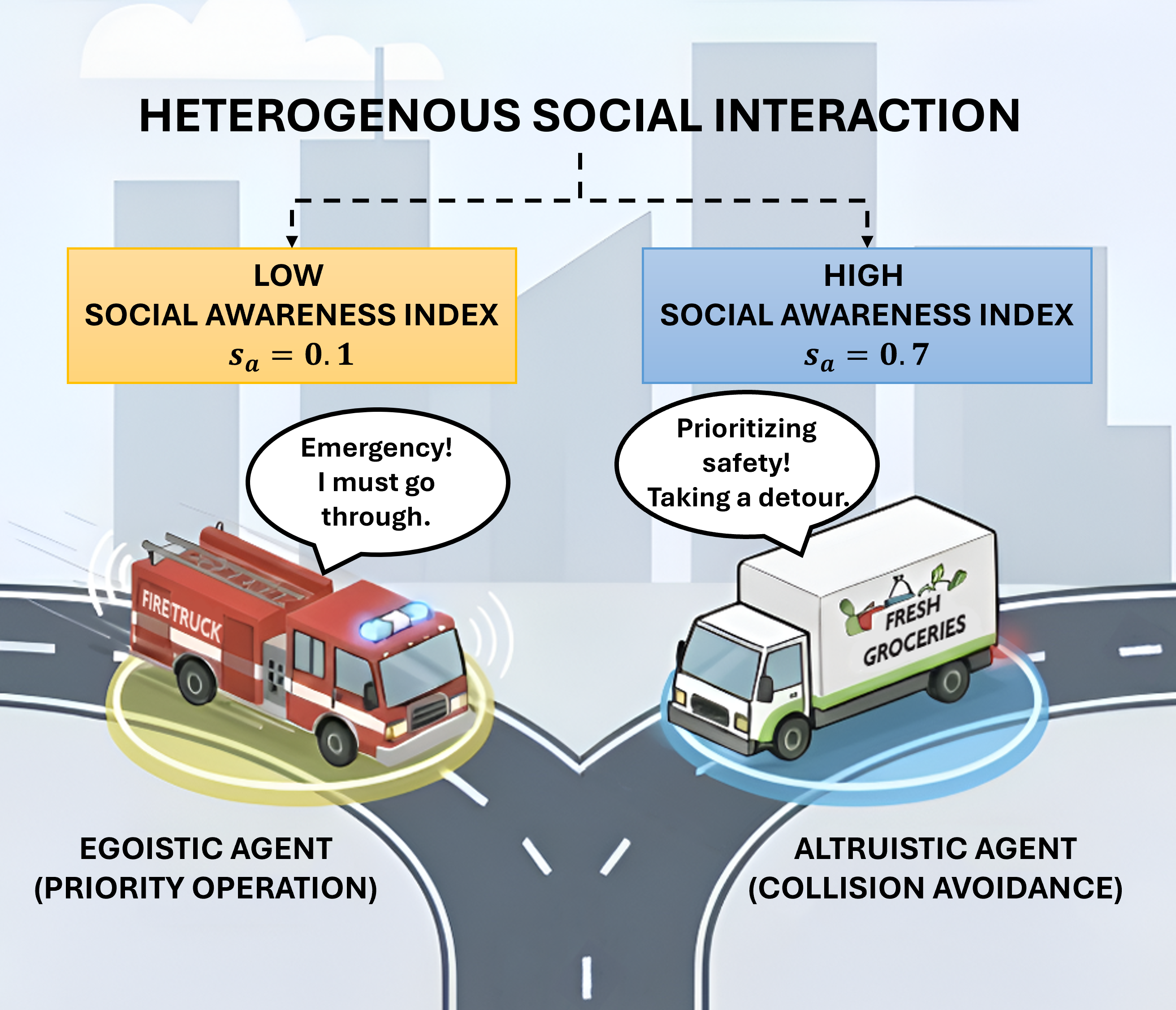}
    \caption{Interaction between an egoistic (high-priority) fire truck and an altruistic (collision-avoiding) grocery vehicle.}
    \label{fig:soc}
\end{figure}

\section{Preliminaries and Problem Formulation}
\label{sec:prelim}
\subsection{Notation}
The symbols $\N$, $\R$, $\R^+$, and $\R_0^+ $ denote the set of natural, real, positive real, and nonnegative real numbers. 
For $a,b\in\R$ and $a< b$, we use $(a,b)$ to represent open interval in $\R$ and $[a,b]$ to represent closed interval in $\R$. For $a,b\in\N$ and $a\leq b$, we use $[a;b]$ to denote close interval in $\N$.
The cardinality of a set \textbf{A} is denoted by $\card{\textbf{A}}$.
Given $N \in \N$ sets $\textbf{A}_i$, $i\in\left[1;N\right]$, we denote the Cartesian product of the sets by $\textbf{A}=\prod_{i\in\left[1;N\right]}\textbf{A}_i:=\{(a_1,\ldots,a_N)|a_i\in \textbf{A}_i,i\in\left[1;N\right]\}$. 
The space of bounded continuous functions is denoted by $\mathcal{C}$. 
A ball centerd at $\cen \in \mathbb{R}^n$ with radius $\rad \in \mathbb{R}^+$ is defined as $\mathcal{B}(\cen, \rad) := \{ x \in \mathbb{R}^n \mid \|x - \cen\| \leq \rad \}$. We use $x\circ y$ to represent the element wise multiplication where $x,y\in \R^n$. We use $I_n$ to represent identity matrix of order $n\in \N$.
All other notations in this paper follow standard mathematical conventions.

\subsection{System Definition}

We consider a Multi-Agent System (MAS) network with a set of agents $\Ag $, where the number of agents is denoted by $n_a:=\card{\Ag } \in \N $. The network is represented as 
\begin{align} \label{eqn:sysdyn}
    \Sigma=\Big\{\Sigma^{(1)},\Sigma^{(2)},\ldots,\Sigma^{(n_a)}\Big\}.
\end{align} 
where the $k^{th}$ agent $\Sigma^{(k)}$ is defined as a control-affine, multi-input multi-output (MIMO), nonlinear strict-feedback system
\begin{align} \label{eqn:sysdyn_ind}
    &\dot{x}_z^{(k)}\!=\!f_z^{(k)}(\overline{x}_z^{(k)})\!+\!g_z^{(k)}(\overline{x}_z^{(k)})x_{z+1}^{(k)}\!+\!w_z^{(k)}, z\in [1;N-1], \notag\\
    &\dot{x}_{N}^{(k)}\!=\!f_N^{(k)}(\overline{x}_N^{(k)})\!+\!g_N^{(k)}(\overline{x}_N^{(k)})u^{(k)}\!+\!w_N^{(k)},\notag\\
    &y^{(k)}\!=\!x_1^{(k)},
\end{align}
where for each $t\in\R^+_0,k \in \Ag$ and $z\in[1;N]$,
\begin{itemize}
    \item $x_z^{(k)}(t)= [x_{z,1}^{(k)}(t), \ldots, x_{z,n}^{(k)}(t)]^\top \in \mathbb{R}^{n}$ is the state,
    \item $\overline{x}_z^{(k)}(t) := [x_1^\top(t),...,x_z^\top(t)]^\top \in  \mathbb{R}^{nz} $,
    \item $u^{(k)}(\overline{x}_z^{(k)},t) \in \mathbb{R}^n$ is control input vector,
    \item $w_z^{(k)}(t) \in \mathbf{W} \subset \R^n$ is unknown bounded disturbance, and
    \item $y^{(k)}(t) = [x_{1,1}^{(k)}(t), \ldots, x_{1,n}^{(k)}(t)]\in \R^n$ is the output.
\end{itemize}

The functions $f_z^{(k)}: \R^{nz}\rightarrow \mathbb{R}^n$ and $g_z^{(k)}:  \R^{nz} \rightarrow \mathbb{R}^{n \times n}$ satisfy the following assumptions:
\begin{assumption}\label{assum:lip}
    For all $k\in\Ag \text{ and }z \in [1;N]$, the functions $f_z^{(k)}$ and $g_z^{(k)}$ are \emph{unknown} but locally Lipschitz continuous.
\end{assumption}
\begin{assum}[\cite{PPC1},\cite{PPC0}] \label{assum:pd}
    For all $\overline{x}_z^{(k)} \in \R^{nz}$, the symmetric part of $g_z^{(k)}$,  which is defined as $g_{z,s}^{(k)}(\overline{x}_z^{(k)}) :=\frac{g_z^{(k)}(\overline{x}_z^{(k)})+g_z^{(k)}(\overline{x}_z^{(k)})^\top}{2}$ is uniformly sign definite with known sign. Without loss of generality, we assume $g_{z,s}^{(k)}(\overline{x}_z^{(k)})$ is positive definite, that is, there exists a constant $\underline{g_z}^{(k)}\in\mathbb R^+, \forall k \in \Ag ,\forall z \in [1;N]$ such that
    $$0 < \underline{g_z}^{(k)} \leq \lambda_{\min} (g_{z,s}^{(k)}(\overline{x}_z^{(k)})), \forall \ \overline{x}_z^{(k)} \in \R^{nz},$$
    where $\lambda_{\min}(\cdot)$ denotes the smallest eigenvalue of a matrix.
\end{assum}
This assumption ensures that global controllability is guaranteed in \eqref{eqn:sysdyn}, i.e., $g_{z,s}^{(k)}(\overline{x}_z^{(k)}) \neq 0,$ for all $\overline{x}_z^{(k)} \in  \R^{nz}$. 
\begin{subsection}{Socially Aware Multi-Agent System (SA-MAS)}
    We modify the MAS definition in \eqref{eqn:sysdyn} to formally incorporate heterogeneous interactions among agents. The extended definition is referred to as a Socially Aware Multi-Agent System (SA-MAS) and is defined as follows:
\begin{align} \label{eqn:social_mas}
    \Sigma_s=\Big\{(\Sigma^{(1)},\soc^{(1)}),(\Sigma^{(2)},\soc^{(2)}),\ldots,(\Sigma^{(n_a)},\soc^{(n_a)})\Big\},
\end{align}
where $\Sigma^{(k)}$ is the system dynamics of each agent as defined in \eqref{eqn:sysdyn_ind} and $\soc^{(k)} \in (0,1)$ is a given \emph{social awareness index} for agent $k \in \Ag $, quantifying its social preference and priority level. 
A smaller value of $\soc^{(k)}$ corresponds to a more \emph{egoistic} agent that prioritizes its own task with less sensitivity towards others, whereas a larger value indicates an \emph{altruistic} agent that is more flexible and cooperative. 
The social awareness index can be assigned based on factors such as task priority, i.e., agents with high priority tasks may have a higher social awareness index. The value of $\soc^{(k)}$ may also reflect an intrinsic characteristic of an agent, such as personality. A more detailed discussion can be found in Section~\ref{sec:sif}.
\end{subsection}

\subsection{Problem Formulation}

\begin{defn}[Temporal Reach Avoid Stay (TRAS)]\label{def:satras}
We say that the \textit{Socially Aware Multi-Agent System \textbf{(SA-MAS)}} $\Sigma_s$ with agents in Equation~\eqref{eqn:sysdyn_ind} satisfies the \textbf{TRAS}
\footnote{The avoid condition collectively ensures that each agent remains collision-free with both dynamic obstacles and other agents at all times.}
task if for each agent $k \in \Ag $, characterized by social awareness index $\soc^{(k)}\in (0,1)$, the following conditions hold 
\begin{align}\label{eqn:satras}
   &y^{(k)}(0) \in \So^{(k)}, \  y(t) \in \T^{(k)}, \forall t \in [t_c^{(k)}, \infty) \\
   &y^{(k)}(t) \notin \U(t), \forall t\in \R_0^+, \\ 
   &y^{(k)}(t)\neq y^{(l)}(t), \forall t\in \R_0^+, \forall l \in\Ag  \setminus \{k\}, 
\end{align}
where $\Ag $ is the set of agents, and $\U: \R_0^+ \rightarrow \R^n$ is a time-varying unsafe set. For each agent $k \in \Ag $,
$t_c^{(k)} \in \R^+$ is the prescribed completion time, $\So ^{(k)}\subset \R^n\setminus \U(0)$ is the initial set, and $\T^{(k)}\subset \R^n \setminus \bigcup_{t\in [t_c^{(k)}, \infty)} \U(t)$ is the  target set. 
\end{defn}

\begin{assum}
    We define the unsafe set $\U(t)\subset \R^n$ as the union of $n_o$ time-varying balls, each centerd around an individual obstacle: 
    $$\U(t) = \bigcup_{j=1}^{n_o} \Obs^{(j)}(t), \text{ where } \Obs^{(j)}(t):= \mathcal{B}(o^{(j)}(t), \rad_o^{(j)}(t)).$$ 
    Here, $\Obs^{(j)}(t)$ is a closed ball of radius $\rad_o^{(j)}(t) \in \R_0^+$ centerd at $o^{(j)}(t) \in \R^n$, capturing the region surrounding the $j^{th}$ dynamic obstacle. Since these regions are defined independently, it allows modeling multiple disconnected and dynamically evolving unsafe regions.
\end{assum}

We now formally state the main control problem addressed in this work.

\begin{prob}[Real-time TRAS]\label{prob1}
Given the SA-MAS, $\Sigma_s$, in Equation~\eqref{eqn:social_mas}, under Assumptions \ref{assum:lip} and \ref{assum:pd}, and a TRAS task as defined in Definition \ref{def:satras}, synthesize a \textit{real-time}, \textit{approximation-free}, and \textit{closed-form} control law $u(t)$ that guarantees the output trajectory $y^{(k)}(t),\forall k \in \Ag $ satisfies the TRAS specification while respecting associated social awareness index.
\end{prob}

To solve this problem, we utilize the STT framework \cite{das2024spatiotemporal}, which defines a time-varying region in the output space that remains safe and goal-directed throughout the horizon.
\begin{defn}[STT for TRAS]\label{STT-defination}
    Given a \textbf{TRAS} task in Definition~\ref{def:satras}, a spatiotemporal tube  (STT), $\Gamma^{(k)}$ for each agent $k\in \Ag $ is defined by 
    $$\Gamma^{(k)}(t)=\B(\cen^{(k)}(t),\rad^{(k)}(t)),$$
    where the tube is characterized by a time varying center $\cen^{(k)}:\mathbb{R}_0^+ \rightarrow\mathbb{R}^n$ and radius $\rad^{(k)}:\R^+_0\rightarrow\R^+$, if the following holds for all $k \in \Ag $:
     \begin{subequations}
        \begin{align}
            &\rad^{(k)}(t) \in\mathbb{R}^+, \forall t \in\R_0^+,\\
            &\Gamma^{(k)}(0) \subset \mathbf{S}^{(k)}, \quad \Gamma(t) \subset \mathbf{T}^{(k)}, \forall t \in [t_c^{(k)},\infty),\\
            &\Gamma^{(k)}(t) \cap \U(t)=\emptyset, \forall t\in\R_0^+,\\
            &\Gamma^{(k)}(t) \cap \Gamma^{(l)}(t)=\emptyset,\forall t\in\R_0^+,\forall l \in \Ag  \setminus \{k\}.    
 \end{align}
    \end{subequations}
\end{defn}

The TRAS specification can be satisfied by designing a control law independently for each agent $k \in \Ag $, such that the output trajectories remain within the corresponding STTs:
\begin{align} \label{eqn:stt_con}
   y^{(k)}(t) \in \Gamma^{(k)}(t), \forall t \in \R_0^+, \forall k\in \Ag.
\end{align}

\section{Designing Spatiotemporal Tubes (STTs)} \label{sec:tube}
In this section, we introduce the construction of spatiotemporal tubes (STTs) for SA-MAS $\Sigma_s$ in Equation~\eqref{eqn:social_mas} to ensure the satisfaction of the Temporal Reach-Avoid-Stay (TRAS) specification in Definition \eqref{def:satras}.
We first develop the design for an arbitrary agent $k\in \Ag $, noting that the same reasoning extends to all other agents. The discussion begins with preliminary definitions and a separation assumption for safety. We then introduce the dynamics of the tube centers and radii and provide the intuition underlying their formulation. Next, we incorporate the social awareness index of each agent into these dynamics to capture heterogeneous agent behaviors. Finally, we present the main theorem and its proof, which formally guarantees that the constructed STTs satisfy the TRAS specification.

\subsection{Preliminaries of STT Design}

We start by selecting the following two points 
\begin{align}
    s^{(k)}=[s^{(k)}_1,...,s_n^{(k)}]\in int(\So^{(k)}),\
    \eta^{(k)}=[\eta^{(k)}_1,...,\eta_n^{(k)}]\in int(\T^{(k)}).\notag
\end{align}
 We define the sets $\hat\So^{(k)}$ and $\hat\T^{(k)}$ as closed balls centerd at the points $s^{(k)}$ and $\eta^{(k)}$, with radii $d_S^{(k)},d_T^{(k)}\in \R^+$, respectively:
\begin{align}
    \hat\So^{(k)} &\!=\!\B(s^{(k)},d_S^{(k)})\!:=\!\{x'\!\in\! \mathbb{R}^n|\norm{x'\!-\!s^{(k)}}\!\leq\! d_S^{(k)}\}\label{eqn:start_state}\\
    \hat\T^{(k)} &\!=\!\B(\eta^{(k)},d_T^{(k)})\!:=\!\{x'\!\in \!\mathbb{R}^n|\norm{x'\!-\!\eta^{(k)}}\!\leq\! d_T^{(k)}\}\label{eqn:targ_state}
\end{align}
such that $\hat \So^{(k)}\subset\So^{(k)}$ and $\hat\T^{(k)}\subset \T^{(k)}$. As introduced in Definition~\ref{STT-defination}, the STT $\Gamma^{(k)}(t) = \mathcal{B}(\cen^{(k)}(t), \rad^{(k)}(t))$ is defined by a time-varying center $\cen^{(k)}: \R_0^+ \rightarrow \R^n$ and radius $\rad^{(k)}: \R_0^+ \rightarrow \R^+$.
Additionally, to ensure a safe approach to the target, we make the following separation assumption:
\begin{assum}\label{ass_pmin}
    For time $t \in[t_c^{(k)},\infty)$, the STT center $\cen(t_c^{(k)})$ is separated from the unsafe set by a known minimum distance of $\rad^{(k)}_{max}\in \R^+$, i.e.,
    $\forall x \in \U (t_c^{(k)}), \|x - \cen^{(k)}(t_c^{(k)})\| >  \rad^{(k)}_{max}$, where $\rad^{(k)}_{max}$ is the maximum allowable tube radius for the $k^{th}$ agent and is choosen such that $\rad^{(k)}_{max}\leq \min(d_S^{(k)},d_T^{(k)})$. 
    Additionally, for all distinct pairs of agents $\{k,l\}\in \Ag ,k\neq l$ the initial set and target set is separated by at least a known minimum distance $\rad^{(k,l)}_{max}=\rad^{(l)}_{max}+\rad_{max}^{(k)}$.
\end{assum}
 These assumptions ensure that agents are not initialized or assigned targets too close to each other or to the unsafe set.
 
 In this decentralized framework, each agent $k \in \Ag$ uses only the information available within its local sensing range. We characterize this information through the time-varying sets of neighboring obstacles $\mathcal{N}_o^{(k)}(t)$ and neighboring agents $\mathcal{N}_a^{(k)}(t)$, defined as:
\begin{align} \label{eqn:neighborhood}
    \mathcal{N}_o^{(k)}(t) &= \{ j \in \{1, \dots, n_o\},  \quad \text{such that,} \| \cen^{(k)}(t) - o^{(j)}(t) \| - \rad_o^{(j)}(t) \leq \rad^{(k)}_{max} \}, \notag \\
    \mathcal{N}_a^{(k)}(t) &= \{ l \in \Ag \setminus \{k\} , \quad \text{such that,} \| \cen^{(k)}(t) - \cen^{(l)}(t) \| \leq \rad^{(k,l)}_{max} \}.
\end{align}
where $\rad^{(k)}_{max}$ acts as a sensing range of the agent $k$.
\subsection{STT center Dynamics}

The evolution of the center $\cen^{(k)}(t)$ of the tube, starting at $\cen^{(k)}(0)=s^{(k)}$, for each agent $k \in \Ag $, is governed by the following dynamics:
\begin{align}\label{eqn:cen_dynamic}
   & \dot{\cen}^{(k)}(t)=\gamma^{(k)}(\cen^{(k)},t) +  \sum_{j \in \mathcal{N}_o^{(k)}(t)}\Big (h_{2,j}^{(k)}m_j^{(k)}(t) + h_{3,j}^{(k)}v_j^{(k)}(t) \Big )  \alpha_j^{(k)}(t) +\\
    &\sum_{l \in \mathcal{N}_a^{(k)}(t)}\Big (\hat h_{2}^{(k,l)}\hat m^{(k,l)}(t) + \hat h_{3}^{(k,l)}\hat v^{(k,l)}(t) \Big ) \beta^{(k,l)}(t)\phi^{(k,l)}(\soc,t).\notag
\end{align}
The first term $\gamma^{(k)}(\cen^{(k)},t)$ is responsible for driving the center trajectory $\cen^{(k)}(t)$ towards the target point  $\eta^{(k)}$, and is defined as follows:
\begin{align}\label{eqn:goal}
\gamma^{(k)}(\cen^{(k)},t) =
\begin{cases}
h_1^{(k)} \dfrac{t_c^{(k)}}{t_c^{(k)} - t}\bigl(\eta^{(k)} - \cen^{(k)}(t)\bigr), & t < t_c^{(k)}, \\[2ex]
0, & t \geq t_c^{(k)}.
\end{cases}
\end{align}
where $h_1^{(k)}>1/t_c^{(k)}$ is a positive constant, dictating the rate of approach towards the goal point $\eta^{(k)}$.

The second term in \eqref{eqn:cen_dynamic} is responsible for collision avoidance with the obstacle and is activated through a switching function $\alpha^{(k)}_j(t)$ when the STT center $\cen^{(k)}(t)$ approaches the $j$-th obstacle: 
\begin{equation*}
    \alpha^{(k)}_j(t)\! =\! 
    \begin{cases}
   \frac{1}{ \| \cen^{(k)}(t)\! -\! o^{(j)}(t) \| \!-\! \rad_o^{(j)}(t)\!}-\frac{1}{\rad^{(k)}_{max}}, &\text{if } \| \cen^{(k)}(t)\! -\! o^{(j)}(t) \| \!-\! \rad_o^{(j)}(t)\! \leq\! \rad^{(k)}_{max}, \\
    0, &\text{otherwise}.
    \end{cases}    
\end{equation*}
The obstacle avoidance term in Equation \eqref{eqn:cen_dynamic} is governed by the positive constants, $h_{2,j}^{(k)},h_{3,j}^{(k)}\in \R^+$, dictating the repulsion strength from the unsafe sets, and the two vectors $m_j^{(k)}(t),v_j^{(k)}(t)\in \R^n$, defined as:
\begin{align}\label{eqn:m_obs}
    m_j^{(k)}(t)\! =\! \frac{\cen^{(k)}(t) - o^{(j)}(t)}{\left( \| \cen^{(k)}(t) \!-\! o^{(j)}(t) \| \!-\! (\rad_o^{(j)}(t)\!+\!\rad^{(k)}_{min})\right)^3}
\end{align}
and $ v_j^{(k)}(t) \in \mathbb{R}^n $ lies in the null space of $m_j^{(k)}(t)$, i.e., $ {m_j^{(k)}}^\top(t) v_j^{(k)}(t) = 0.$ $\rad^{(k)}_{min} \in \R^+$ is the minimum tube radius.  So whenever the STT center approaches the $j$-th unsafe set the switching function gets activated, i.e., $\alpha^{(k)}_j(t)\neq0$ and the vector  $m_j^{(k)}(t)$ in \eqref{eqn:m_obs} together with its orthogonal component $v_j^{(k)}(t)$ steer the tube around the unsafe set, ensuring safety. 

The third term addresses inter-agent collision avoidance between agent $k$ and its neighboring agent $l \in \mathcal{N}^{(k)}_a(t)$. It is activated through a switching function $\beta^{(k,l)}$ when the STT center of the $k$-th agent $\cen^{(k)}(t)$ approaches the STT center of the $l$-th agent:
\begin{equation}\label{eqn:switch_agent}
    \beta^{(k,l)}(t) = 
    \begin{cases}
    \frac{1}{\| \cen^{(k)}(t) - \cen^{(l)}(t) \|}-\frac{1}{\rad_{max}^{(k,l)}},& \text{if } \| \cen^{(k)}(t) - \cen^{(l)}(t) \| \leq \rad_{max}^{(k,l)}, \\
    0, & \text{otherwise},
    \end{cases}    
\end{equation}
where $\rad_{max}^{(k,l)}=\rad_{max}^{(k)}+\rad_{max}^{(l)}$.

Once the switching function is activated, the inter-agent collision avoidance is mainly governed by the arbitrary positive constants $\hat h_{2}^{(k,l)},\hat h_{3}^{(k,l)} \in \R^+$ dictating the repulsion rates from neighboring agent, and the two vectors $\hat m^{(k,l)}(t),\hat v{(k,l)}(t)\in \R^n$, defined as:
\begin{align}\label{eqn:m_agent}
 \hat m^{(k,l)}(t)\! =\! \frac{\cen^{(k)}(t) - \cen^{(l)}(t)}{\left( \| \cen^{(k)}(t) \!-\! \cen^{(l)}(t) \| \!-\! (\rad_{min}^{(l)}\!+\!\rad^{(k)}_{min})\right)^3}, 
\end{align}
and $ \hat v^{(k,l)}(t) \in \mathbb{R}^n $ lies in the null space of $\hat m^{(k,l)}(t)$, i.e., $  \left(\hat{m}^{(k,l)}(t)\right)^{\top} \hat v^{(k,l)}(t) = 0.$ 

The function $\phi^{k,l}(\soc, t)$ is the Social Interaction Function (SIF), determined by the social awareness indices $\soc = (\soc^{(k)}, \soc^{(l)})$ of agents $k$ and $l$, as defined in Subsection \ref{sec:sif}.

\subsection{STT Radius Dynamics}

The tube radius $\rad^{(k)}(t)$ is dynamically adjusted based on the proximity to local neighbors defined in \eqref{eqn:neighborhood}:
\begin{equation}\label{eqn:radius_dynamics}
\dot{\rad}^{(k)}(t)=\frac{\mathsf{e}^{-\nu d_1^{(k)}{(t)}}\dot{ d_1}^{(k)}(t)+\mathsf{e}^{-\nu d_2{(t)}}\dot{ d_2}^{(k)}(t)}{\big ( \mathsf{e}^{-\nu\rad^{(k)}_{max}}+\mathsf{e}^{-\nu d_1^{(k)}(t)}+\mathsf{e}^{-\nu d_2^{(k)}(t)}\big)},
\end{equation}
where $\nu \in \R^+$ is an arbitrary smoothening parameter, and 
$d^{(k)}_1(t)$, $d^{(k)}_2(t)$ represent the smooth minimum distances to the sets $\mathcal{N}_o^{(k)}(t)$, $\mathcal{N}_a^{(k)}(t)$ in \eqref{eqn:neighborhood}, respectively:
\begin{align}
     d_1^{(k)}(t) &= -\frac{1}{\nu}\ln \Big(\sum_{j \in \mathcal{N}_o^{(k)}(t)}\mathsf{e}^{-\nu  d'^{(k)}_j(t)} \Big) \label{eqn:rho} \\
    d_2^{(k)}(t) &= -\frac{1}{\nu}\ln  \Big( \sum_{l \in \mathcal{N}_a^{(k)}(t)} \mathsf{e}^{-\nu  d'^{(k,l)}(t)} \Big)\label{eqn:rad_agent}
\end{align}
and the corresponding time derivatives are given by $\dot d_1^{(k)}(t)$ and $\dot d_2^{(k)}(t)$, by:
\begin{align*}
    \dot d_1^{(k)}(t) &= \frac{\sum_{j \in \mathcal{N}_o^{(k)}(t)} \mathsf{e}^{-\nu d_j'^{(k)}(t)} \, \dot d_j'^{(k)}(t)}{\sum_{j \in \mathcal{N}_o^{(k)}(t)} \mathsf{e}^{-\nu d_j'^{(k)}(t)}}, \
    \dot d_2^{(k)}(t) = \frac{\sum_{l \in \mathcal{N}_a^{(k)}(t)} \mathsf{e}^{-\nu d'^{(k,l)}(t)} \, \dot d'^{(k,l)}(t)}{\sum_{l \in \mathcal{N}_a^{(k)}(t)} \mathsf{e}^{-\nu d'^{(k,l)}(t)}}.
\end{align*}
In above equations $d^{(k)}_1(t)$ is the smooth minimum of the distances between the tube center $\cen^{(k)}(t)$ and each locally sensed obstacle, defined as $d'^{(k)}_j(t)=\| \cen^{(k)}(t) - o^{(j)}(t) \| - \rad_o^{(j)}(t),$ over all $j \in \mathcal{N}_o^{(k)}(t)$. $ d_2^{(k)}(t)$ is the smooth minimum over all locally sensed neighboring agents of $d'^{(k,l)}(t)=\rad_{min}^{(k)}+\Big(\| \cen^{(k)}(t) - \cen^{(l)}(t) \| - \rad_{min}^{(k,l)}\Big)W(\phi^{(k,l)},t),$ 
where $W(\phi^{(k,l)},t)=\Big(1-\frac{\| \cen^{(k)}(t) - \cen^{(l)}(t) \|-\rad_{min}^{(k,l)}}{\rad_{max}^{(k,l)}-\rad_{min}^{(k,l)}}\Big)\big(1-\phi^{(k,l)}(\soc,t)\big)+\Big(\frac{\| \cen^{(k)}(t) - \cen^{(l)}(t) \|-\rad_{min}^{(k,l)}}{\rad_{max}^{(k,l)}-\rad_{min}^{(k,l)}}\Big)\Big(\frac{\rad_{max}^{(k)}-\rad_{min}^{(k)}}{\rad_{max}^{(k,l)}-\rad_{min}^{(k,l)}}\Big)$which represents the distance between centers of the tubes of two neighboring agents, adjusted by the sum of their minimum allowable tube radii and weighted by the social interaction function $\phi^{(k,l)}(\soc,t)$.

The radius $\rad^{(k)}(t)$ of the tube changes in two cases. First, when the center of the tube is close to any obstacle, the radius shrinks to avoid collision with the unsafe set, and expands when it is farther away. Second, the radius adapts according to the social indices of agent $k$ $\soc^{(k)}$, and its neighbour $l$ $\soc^{(l)}$, whenever the agents are in close proximity.

Thus, given a time-varying center $\cen^{(k)}: \R_0^+ \rightarrow \R^n$ and radius $\rad^{(k)}: \R_0^+ \rightarrow \R^+$ for each agent $k \in \Ag $, governed by the dynamics in Equations \eqref{eqn:cen_dynamic} and \eqref{eqn:radius_dynamics}, we define the STT $\Gamma^{(k)}(t) = \mathcal{B} (\cen^{(k)}(t), \rad^{(k)}(t))$ as a closed ball in $\R^n$ centerd at $\cen^{(k)}(t)$ with radius $\rad^{(k)}(t)$:
\begin{equation}\label{eqn:stt_ball}
    \Gamma^{(k)}(t)\! :=\! \{ x\!\in \!\R^n \!\mid \! \|x\!-\!\cen^{(k)}(t)\| \!\leq\! \rad^{(k)}(t)\}, \forall t \!\in \!\R_0^+ .
\end{equation}

\subsection{Social Interaction Function (SIF)}\label{sec:sif}
For a given SA-MAS ($\Sigma_s$) with a social awareness index $\soc^{(k)},\forall k \in \Ag$, the value of $\soc^{(k)}$ determines the interaction between each agent while solving the assigned TRAS task. To incorporate these social behaviors into the STT, we define the Social Interaction Function (SIF) $\phi^{(k,l)}(\soc, t)$ as follows:
\begin{equation}\label{eqn:SIF}
\phi^{(k,l)}(\soc,t)\! =\!
\begin{cases}
\frac{\soc^{(k)}}{\soc^{(l)}+\soc^{(k)}}, & t < t_c^{(k)}, \\
\frac{\soc^{(k)}}{\soc^{(l)}+\soc^{(k)}} \mathsf{e}^{\big(-\frac{(t-t_c^{(k)})^2}{b^2}\big)}, & t \ge t_c^{(k)},
\end{cases}
\end{equation}
where $b \in [0,1]$. A higher social awareness index $\soc^{(k)} \gg \soc^{(l)}$ yields $\phi^{(k,l)}(\soc,t) \approx 1$, representing an \textit{altruistic} agent willing to compromise its task to avoid collisions. In contrast, a lower social awareness index $\soc^{(k)} \ll \soc^{(l)}$ gives $\phi^{(k,l)}(\soc,t) \approx 0$, representing an \textit{egoistic} agent prioritizing its own TRAS task.

When agent $k$ approaches neighboring agent $l$, two things happen to tube $\Gamma^{(k)}(t)$. The tube center $\cen^{(k)}(t)$ steers around the neighboring agent, and the radius $\rad^{(k)}(t)$ shrinks to ensure that the intersection of the two tubes is empty at all times. It is important to note that similar adaptations also occur for tube $\Gamma^{(l)}(t)$, as $\cen^{(l)}(t)$ steers around tube $k$ and $\rad^{(l)}(t)$ shrinks to maintain disjoint tubes. Now, the relative effort for collision avoidance is influenced by the social indices of the two agents $\soc^{(k)}$ and $\soc^{(l)}$, through the weighting functions $\phi^{(k,l)}(\soc,t)$ and $\phi^{(l,k)}(\soc,t)$.
\begin{itemize}
    \item An agent with lower social awareness index $\soc^{(k)}$ yields lower $\phi^{(k,l)}(\soc,t)$, and therefore gives less weight to the third term in \eqref{eqn:cen_dynamic} and $d'^{(k,l)}(t)$ in Equation~\eqref{eqn:rad_agent}. So, its tubes bend and shrink minimally, prioritizing its own task.
    \item An agent with larger social awareness index $\soc^{(k)}$ yields larger $\phi^{(k,l)}(\soc,t)$, and therefore assigns more weight to the third term in \eqref{eqn:cen_dynamic} and $d'^{(k,l)}(t)$ in Equation~\eqref{eqn:rad_agent}. Thus, their tubes bend and shrink more, prioritizing inter-agent collision-avoidance.
    \item If a pair of agents share similar social indices $\soc^{(k)} \approx \soc^{(l)}$, then irrespective of whether it is high or low, the tube for both the agents are considered equally responsible for avoiding inter-agent collisions $\phi^{(k,l)}(\soc,t) \approx \phi^{(l,k)}(\soc,t) \approx 0.5$. This ensures a fair and symmetric treatment of agents with similar social behaviors.
\end{itemize}

Finally, the parameter $b$ in Equation~\eqref{eqn:SIF} controls the rate at which $\phi^{(k,l)}(\soc,t)$ decays to zero after agent $k$ reaches its target at $t_c^{(k)}$. This design choice is made to ensure that, once an agent reaches its target, it prioritizes remaining within the target region thereafter, consistent with the TRAS specification in Definition~\ref{def:satras}. As $\phi^{(k,l)}(\soc,t) \approx 0$ after arrival, agent $k$ behaves more egoistically and no longer actively adapts its motion for neighboring agents. In the current framework, this does not compromise safety since inter-agent collision avoidance is still guaranteed through the non-intersection of the STTs. However, in highly congested goal regions or narrow passages, the post-arrival interaction rule may be relaxed depending on the task objective.

\subsection{Theoretical Guarantee of TRAS Satisfaction}

The next theorem guarantees that the designed STT in Equation~\eqref{eqn:stt_ball} adheres to the key conditions for satisfying TRAS specifications in \eqref{eqn:satras}. 
\begin{thm}\label{theorem_tube}
The STT $\Gamma^{(k)}(t),\forall k \in \Ag $ in \eqref{eqn:stt_ball} meets the following to ensure satisfaction of the TRAS specification:
\begin{enumerate}
    \item[(i)] The tubes for each agent $k \in \Ag $ reach their respective target within the prescribed time $t_c$ and stays within it thereafter: $\Gamma^{(k)}(t) \subseteq \T^{(k)},\forall t\in[t_c^{(k)},\infty)$.
    \item[(ii)] The tubes for each agent $k \in \Ag $ avoid the unsafe set at all times: 
    $\Gamma^{(k)}(t) \cap \U(t) = \emptyset$,  $\forall t \in \R_0^+$.
     \item[(iii)] The tubes of any two distinct agents do not intersect, regardless of their social awareness indices, i.e.,
\begin{align}
    \Gamma^{(k)}(t) \cap\Gamma^{({l})}(t) = \emptyset, 
     \forall t \in \R_0^+,
     \forall \{k,l\}\in \Ag ,k\neq l\notag.
\end{align}
     
    \item[(iv)] The STT radius for each agent $k \in \Ag $ remains  positive throughout the motion: $\rad^{(k)}(t) \in \mathbb{R}^+, $ $\forall t \in \R_0^+$.
\end{enumerate}
\end{thm}

\begin{pf}
We prove each claim individually:

(i) By Assumption~\ref{ass_pmin}, at $t = t_c^{(k)}$, the tube of agent $k$ is sufficiently separated from all unsafe sets, and the centers of all other agents $l \in \mathcal{N}^{(k)}_a(t)$. Therefore, we have $\alpha^{(k)}_j(t_c^{(k)}) = \beta^{(k,l)}(t_c^{(k)}) = 0$ for all $j\in \mathcal{N}^{(k)}_o(t)$ and for all $l \in \Ag \setminus \{k\}$. Substituting this in Equation~\eqref{eqn:cen_dynamic}, we get:
\begin{align*}
    \dot\cen^{(k)}(t)=h_1^{(k)}\frac{t_c^{(k)}}{t_c^{(k)}-t}(\eta^{(k)}-\cen^{(k)}(t)).
\end{align*}
Solving this equation, we obtain $\cen^{(k)}(t)=\eta^{(k)}+C(t_c^{(k)}-t)^{h_1^{(k)}t_c^{(k)}}$, where $C$ is a constant determined using the initial condition $\cen^{(k)}(0)=s^{(k)}$. The solution here approaches $\eta^{(k)}$ as $t$ approach $t_c^{(k)}$, i.e., $\cen^{(k)}(t_c^{(k)})=\eta^{(k)}$, with convergence rate determined by $h_1^{(k)}.$

Next, we can write the closed-form solution for the STT radius dynamics in Equation~\eqref{eqn:radius_dynamics} as follows:
 \begin{align}\label{eqn:rad_sol}
      \rad^{(k)}(t)=\frac{-1}{\nu}\ln(\mathsf{e}^{-v\rad_{max}^{(k)}}+\mathsf{e}^{-vd^{(k)}_1}+\mathsf{e}^{-vd^{(k)}_2}),
  \end{align}
  where $d^{(k)}_1$ and $d^{(k)}_2$ is defined in \eqref{eqn:rho} and \eqref{eqn:rad_agent}. Therefore,
  \begin{align}\label{eqn:rad_min_ineq}
      \rad^{(k)}(t)\leq \min\Big(&\min_{j\in \mathcal{N}^{(k)}_o(t)}(\norm{\cen^{(k)}(t)-o^{(j)}(t)}-\rad_o^{(j)} ),\rad^{(k)}_{max}, \ \min_{l\in \mathcal{N}^{(k)}_a(t)}d'^{(k,l)}\Big),
  \end{align}
  Thus, the radius satisfies $\rad^{(k)}(t) \leq \rad^{(k)}_{max} \leq d_T^{(k)}$ for all $t \in [0,t_c^{(k)}]$. At $t=t_c^{(k)}$, we have $\cen^{(k)}(t_c^{(k)})=\eta^{(k)}$ and hence, 
  \begin{align*}
      \Gamma^{(k)}(t_c^{(k)})&=\B(c^{(k)}(t_c^{(k)}),\rad^{(k)}(t_c^{(k)}))\subseteq\B(\eta^{(k)},d_T^{(k)})= \hat\T^{(k)}\subseteq \T^{(k)}.
  \end{align*}
So far, we have established that $\Gamma^{(k)}(t_c) \subset \T^{(k)}$. This result can be extended to all subsequent times $t\in(t_c^{(k)},\infty)$ by using \eqref{eqn:goal} and Assumption \ref{ass_pmin}. Under these conditions, the first and second terms in \eqref{eqn:cen_dynamic} become zero, while the third term also vanishes due to the definition of the SIF in \eqref{eqn:SIF}, i.e., $\phi^{(k,l)}(\soc,t) = 0$. Consequently, we obtain $\dot{\cen}^{(k)} = 0$ for all $t \in (t_c^{(k)}, \infty)$, which implies that $\cen^{(k)}(t) = \eta^{(k)}$ over the same time interval.

Next, for radius $\rad^{(k)}(t)$ of STT we can extend the argument that $\rad^{(k)}(t) \leq \rad^{(k)}_{max} \leq d_T^{(k)}$ for all $t \in (t_c^{(k)},\infty)$, hence, we conclude that:
  \begin{align*}
      \Gamma^{(k)}(t)&=\B(c^{(k)}(t),\rad^{(k)}(t)) \subseteq\B(\eta^{(k)},d_T^{(k)})= \hat\T^{(k)}\subseteq \T^{(k)},\forall t \in [t_c^{(k)},\infty).
  \end{align*}
(ii)  We prove the second claim for an arbitrary agent $k \in \Ag $, which extends to all agents. For $j \notin \mathcal{N}^{(k)}_o(t)$, $\Gamma(t) \cap \Obs^{(j)}(t) = \emptyset$. Now, for each unsafe set $j \in \mathcal{N}^{(k)}_o(t)$, we define the following time varying continuous function:
\begin{align}
   J^{(k)}_j(t)=(\cen^{(k)}(t)-o^{(j)}(t))^\top(\cen^{(k)}(t)-o^{(j)}(t))
    -(\rad^{(j)}_o(t)+\rad^{(k)}_{min})^2,\nonumber
\end{align}
which measures the squared distance between the $k^{th}$ agent's tube and the center of the $j^{th}$ unsafe set, offset by $\rad_o^{(j)}+\rad_{min}^{(k)}$, which acts as safety margin. The time derivative of $J^{(k)}_j(t)$ can be written as:
\begin{align}
       \dot J^{(k)}_j(t)=2(\cen^{(k)}(t)-o^{(j)}(t))^\top(\dot \cen^{(k)}(t)-\dot o^{(j)}(t))-2(\rad_o^{(j)}+\rad_{min})\dot\rad^{(j)}(t).\notag
\end{align}
We now look at $\dot{J}^{(k)}_j(t)$ on the boundary of the safe margin around each obstacle, $\|\cen^{(k)}(t) - o^{(j)}(t)\| = \rad_o^{(j)}(t) + \rad^{(k)}_{min}$. Substituting $\cen^{(k)}(t)$ into the expression for $\dot J^{(k)}_j(t)$:
\begin{align*}
     &\dot{J}^{(k)}_j(t) =2h_1^{(k)}\frac{t_c^{(k)}}{t_c^{(k)}-t}(\cen^{(k)}(t)-o^{(j)}(t))^\top(\eta^{(k)}-\cen^{(k)}(t)) \notag \\
     &+2h_{2,j}^{(k)}   \alpha^{(k)}_j(t)\frac{\norm{\cen^{(k)}(t) - o^{(j)}(t)}^2}{\left( \| \cen^{(k)}(t) - o^{(j)}(t) \| - (\rad_o^{(j)}(t)+\rad^{(k)}_{min})\right)^3} \notag \\
     &+ 2h^{(k)}_{3,j}   \alpha^{(k)}_j(t)(\cen^{(k)}(t)-o^{(j)}(t))^\top m^{(k)}_j(t) \ \\
     &+\sum_{l \in \mathcal{N}^{(k)}_a(t)}2 \beta^{(k,l)}(t)\Big (\hat h_{2}^{(k,l)}\hat m^{(k,l)}(t) + \hat h_{3}^{(k,l)}\hat v^{(k,l)}(t) \Big ) \phi^{(k,l)}(\soc,t)\notag\\
     &- 2(\cen^{(k)}(t)-o^{(j)}(t)) \dot{o}^{(j)}(t)-2(\rad_o^{(j)}+\rad^{(k)}_{min})\dot\rad^{(j)}(t). \notag
\end{align*}
where $\alpha^{(k)}_j(t)\in \R^+$  since $\rad^{(k)}_{max}>\rad^{(k)}_{min}$.

As $\norm{ \cen^{(k)}(t) - o^{(j)}(t) } \rightarrow \rad_o^{(j)}(t) + \rad^{(k)}_{min}$, the denominator in the second term approaches zero, making this term dominant and positive. As a result, $\dot{J}^{(k)}_j(t) > 0$ near the boundary. The STT center is initially at a safe distance from the $j$-th unsafe set, i.e., $\norm{ \cen^{(k)}(0) - o^{(j)}(0) } > \rad_o^{(j)}(0) + \rad^{(k)}_{min}$, which implies $J^{(k)}_j(0) > 0$. Since $\dot{J}^{(k)}_j(t) > 0$ as $\norm{ \cen^{(k)}(t) - o^{(j)}(t) } \rightarrow \rad_o^{(j)}(t) + \rad^{(k)}_{min}$, the function $J^{(k)}_j(t)$ cannot decrease to zero or become negative. Therefore, $J^{(k)}_j(t) > 0$ holds for all $t \in [0, t_c^{(k)}]$, implying 
$$\| \cen^{(k)}(t) - o^{(j)}(t) \| > \rad_o^{(j)}(t) + \rad^{(k)}_{min} \ \forall \ t \in [0,t^{(k)}_c].$$
Using the Assumption \ref{ass_pmin}, we can further extend the above results for all subsequent times. Hence, the STT center remains at least $\rad_o^{(j)}(t)+\rad_{min}$ away from the center of the $j$-th obstacle at all times. 
Now, to guarantee that the tube $\Gamma^{(k)}(t)=\mathcal{B}(\cen^{(k)}(t),\rad^{(k)}(t))$ does not intersect with the unsafe set, it suffices to show that:
\begin{align}\label{eqn:radprove}
      \rad^{(k)}(t)\leq d'^{(k)}_j,\forall j \in \mathcal{N}^{(k)}_o(t),
\end{align} 
where $d'^{(k)}_j$ is defined in \eqref{eqn:rho}. We verify this using the solution of the radius dynamics given in Equation~\eqref{eqn:radius_dynamics}. Consider the following two cases:
\\

\textbf{Case 1:} $\forall j \in \mathcal{N}^{(k)}_o(t),\min\Big(\rad^{(k)}_{max},\min_{l\in \mathcal{N}^{(k)}_a(t)}(d'^{(k,l)})\Big) \leq d'^{(k)}_j$:
 Substituting the inequality  in \eqref{eqn:rad_min_ineq} we have:
\begin{align*}
     \rad^{(k)}(t)\leq \min\Big(\rad^{(k)}_{max},\min_{l\in \mathcal{N}^{(k)}_a(t)}(d'^{(k,l)})\Big) \leq d'^{(k)}_j, 
\end{align*}
for all $j \in \mathcal{N}^{(k)}_o(t)$, satisfying condition~\eqref{eqn:radprove}.

\textbf{Case 2:} $\exists \hat j \in \mathcal{N}^{(k)}_o(t),\min\Big(\rad^{(k)}_{max},\min_{l\in \mathcal{N}^{(k)}_a(t)}(d'^{(k,l)})\Big) > d'^{(k)}_{\hat j}$:
 Substituting the inequality in \eqref{eqn:rad_min_ineq} we directly get
     $\rad^{(k)}(t)\leq d'^{(k)}_{\hat j}$
ensuring condition~\eqref{eqn:radprove} holds.

Thus, in both scenarios, \eqref{eqn:radprove} is satisfied. Now, repeating this argument for all $j \in [1 ;n_o]$ shows that the STT center for the $k^{th}$ agent $\Gamma^{(k)}(t)$ does not intersect any unsafe set at any time. The same reasoning applies to all other agents:
$$\Gamma^{(k)}(t) \cap \U(t) = \emptyset, \forall t \in \R_0^+, \forall k \in \Ag .$$ 

It is important to note that the time derivatives of the obstacle position and size are used only in the analysis and are not required for implementation, which relies only on the current obstacle position obtained online.

(iii)  For any two distinct pair of agents $k,l\in \Ag $, if $l \notin \mathcal{N}^{(k)}_a(t)$, then $\Gamma^{(k)}(t) \cap \Gamma^{(l)}(t) = \emptyset.$
Now for $l \in \mathcal{N}^{(k)}_a(t)$ we define a time varying functions:
\begin{align*}
    M^{(k,l)}=(\cen^{(k)}(t)-\cen^{(l)}(t))^\top(\cen^{(k)}(t)-\cen^{(l)}(t))-{\rad_{min}^{(k,l)}}^2,
\end{align*}
which measures the squared distance between the tube center of agents $k$ and $l$, offset by the minimum safety distance $\rad_{min}^{(k,l)} := \rad_{min}^{(k)} + \rad_{min}^{(l)}$. Its time derivative is
\begin{align}
    \dot M^{(k,l)}=2(\cen^{(k)}(t)-\cen^{(l)}(t))^\top(\dot\cen^{(k)}(t)-\dot\cen^{(l)}(t)).
\end{align}
We analyze $\dot M^{(k,l)}$ on the boundary of the safe margin, i.e, when $\norm{\cen^{(k)}(t)-\cen^{(l)}(t)}=(\rad_{min}^{(k)}+\rad_{min}^{(l)})$.
Substituting the dynamics of $\dot\cen^{(k)}(t)$ and $\dot\cen^{(l)}(t)$ from \eqref{eqn:cen_dynamic}, we get:
\begin{align}
    &\dot M^{(k,l)}=2(\cen^{(k)}(t)-\cen^{(l)}(t))^\top\theta^{(k,l)}+2\beta^{(k,l)}\notag\\ &\phi^{(k,l)}(\soc,t) \hat h_{2}^{(k,l)}\frac{\norm{\cen^{(k)}(t) - \cen^{(l)}(t)}^2}{\left( \| \cen^{(k)}(t) - \cen^{(l)}(t) \| - (\rad_{min}^{(l)}+\rad^{(k)}_{min})\right)^3}\notag\\
    &+2\beta^{(k,l)}\phi^{(k,l)}(\soc,t)(\cen^{(k)}-\cen^{(l)})^\top \hat h_{3}^{(k,l)}\hat v^{(k,l)}(t)-2(\cen^{(k)}(t)-\cen^{(l)}(t))^\top \theta^{(l,k)}+2 \beta^{(l,k)}\phi^{(l,k)}(\soc,t) \notag\\
    &\hat h_{2}^{(l,k)}\frac{\norm{\cen^{(k)}(t) - \cen^{(l)}(t)}^2}{\left( \| \cen^{(k)}(t) - \cen^{(l)}(t) \| - (\rad_{min}^{(k)}+\rad^{(l)}_{min})\right)^3}-2\beta^{(l,k)}\phi^{(l,k)}(\soc,t)(\cen^{(k)}(t)-\cen^{(l)}(t))^\top \hat h_{3}^{(l,k)}\hat v^{(l,k)}(t),\notag
\end{align}
where $\theta^{(k,l)}=\sum_{j\in \mathcal{N}^{(k)}_o(t)}\Big (h_{2,j}^{(k)}m_j^{(k)}(t) + h_{3,j}^{(k)}v_j^{(k)}(t) \Big )  \alpha_j^{(k)}(t)+h_1^{(k)}\frac{t_c^{(k)}}{t_c^{(k)}-t}(\eta^{(k)}-\cen^{(k)}(t)$ and $\beta^{(k,l)}=\beta^{(l,k)}\in \R^+$ by its definition in \eqref{eqn:switch_agent}. 

Next, without loss of generality, we assume that for any unique pair of agents $k$ and $l \in \mathcal{N}^{(k)}_a(t)$, 
$t_c^{(k)}\leq t_c^{(l)}$. We now divide the analysis into three cases: 

\textbf{Case 1:} When $t\in [0,t_c^{(k)}]$:
At the boundary, when $\|\cen^{(k)}(t)-\cen^{(l)}(t)\| \rightarrow \rad_{min}^{(k)}+\rad_{min}^{(l)}$, all terms remain bounded except the second and fourth. Since all constants are positive and $\phi^{(k,l)} \in \mathbb{R}^+$ by definition~\eqref{eqn:SIF}, both the second and fourth terms are positive, and their denominators approach zero, making them dominant. Consequently, $\dot{M}^{(k,l)} > 0$ holds for all $t \in [0, t_c^{(k)}]$.

\textbf{Case 2:} When $t\in (t_c^{(k)},t_c^{(l)}]:$
By definition of the SIF function $\phi^{(k,l)}(\soc,t)=0, \forall  t\in (t_c^{(k)},t_c^{(l)}]$. Thus, at the boundary, when $\|\cen^{(k)}(t)-\cen^{(l)}(t)\| \rightarrow \rad_{min}^{(k)}+\rad_{min}^{(l)}$, similar to Case 1 all terms except the fourth remain bounded including the second term which will be equal to zero. Since all constants are positive and $\phi^{(l,k)} \in \mathbb{R}^+$ by definition~\eqref{eqn:SIF}, the fourth term is positive and its denominators approach zero, making it dominant. Consequently, $\dot{M}^{(k,l)} > 0$ holds for all $t \in(t_c^{(k)},t_c^{(l)}]$.

\textbf{Case 3:} When $t\in (t_c^{(l)},\infty)$:
Using the proof of the first statement presented in the theorem, each tube reaches its target set within the prescribed time. With Assumption~\ref{ass_pmin}, it guarantees a minimum safe separation between the targets of all agents $\norm{\cen^{(k)}(t)-\cen^{(l)}(t)}>\rad_{min}^{(k)}+\rad_{min}^{(l)} \implies M^{(k,l)}(t)>0, \forall {t\in(t_c^{(l)},\infty)}$.

Thus, if the tube centers of agents $k$ and $l$ are initially at a safe distance, i.e., $\norm{\cen^{(k)}(0)-\cen^{(l)}(0)}>\rad_{min}^{(k)}+\rad_{min}^{(l)}$, then $M^{(k,l)}>0$. 
Since for $t \in [0,t_c^{(l)}]$, $\dot M^{(k,l)}(t) > 0$, the function $M^{(k,l)}(t)$ cannot decrease to zero or become negative. Hence, $M^{(k,l)}(t)>0$ holds for all $t\in \R_0^+$ and for all social awareness index $\soc^{(k)},\soc^{(l)}\in (0,1)$, i.e.,
\begin{align}
    \norm{\cen^{(k)}(t)-\cen^{(l)}(t)}>\rad_{min}^{(k)}+\rad_{min}^{(l)}, &\forall {t\in\R_0^+}.
\end{align}
Hence, the centers of the tubes of the agents $k$ and $l$ maintain a minimum safety distance. 

Now in order to prove that tubes for agents $k$ and $l$ do not intersect with each other, i.e, $\Gamma^{(k)}(t)\cap\Gamma^{(l)}(t)=\emptyset$ it suffices to show that:
 \begin{align}\label{eqn:rad_cond}
     \rad^{(k)}, \rad^{(l)}\leq d'^{(k,l)}(t),
 \end{align}
 where $d'^{(k,l)}(t)$ is defined in \eqref{eqn:rad_agent}. Now, we consider the two cases for agent $k$ and extend the same reasoning to agent $l$:

  \textbf{Case 1:} $\Big(\rad_{max}^{(k)},\min_{j \in \mathcal{N}^{(k)}_o(t)}\big( d'^{(k)}_j(t)\big)\Big)\leq d'^{(k,l)}(t)$ $\forall l \in \mathcal{N}^{(k)}_a(t)$.
 Substituting the inequality into \eqref{eqn:rad_min_ineq}, we obtain:
 \begin{align*}
     \rad^{(k)}(t)\!\leq\! \min\Big(\min_{j\in \mathcal{N}^{(k)}_o(t)}(d'^{(k)}_j(t) ),\rad^{(k)}_{max}\Big) \!\leq\! d'^{(k,l)}(t).
 \end{align*}

 \textbf{Case 2:}
  $\min\Big(\rad_{max}^{(k)}, \min_{j \in \mathcal{N}^{(k)}_o(t)}\big(d'^{(k)}_j(t)  \big)\Big)>d'^{(k,\hat l)}(t)$ for any $\hat l \in \mathcal{N}^{(k)}_a(t)$. Substituting the inequality into \eqref{eqn:rad_min_ineq} we directly get $\rad^{(k)}(t)\leq d'^{(k,\hat l)}(t)$.
  
 Thus, in both cases, we obtain $\rad^{(k)} \leq d'^{(k,l)}(t)$. Now, applying the same arguments to agent $l$, we show that condition \eqref{eqn:rad_cond} is satisfied.
 
By repeating the same reasoning for any pair of agents $k,l \in \Ag $, we conclude that tubes of all distinct agent pairs do not intersect, regardless of their social awareness indices:
 \begin{align*}
     \Gamma^{(k)}(t)\cap\Gamma^{(l)}(t)=\emptyset, \ &\forall {t\in\R_0^+}, \forall \soc^{(k)},\soc^{(l)}\in (0,1).
 \end{align*}
 (iv) From parts (ii) and (iii) of the proof for each of the agents $k \in \Ag $, we have established the following conditions:
 \begin{align*}
  &   \norm{\cen^{(k)}(t)-o^{(j)}(t)}>\rad_o^{(j)}+\rad^{(k)}_{min},\forall j \in [1;n_o],\\
  &   \norm{\cen^{(k)}(t)-\cen^{(l)}(t)}>\rad^{(k)}_{min}+\rad_{min}^{(l)},\forall l \in \Ag \setminus \{k\},
 \end{align*}
 implying that $d^{(k)}_1\geq \rad_{min}^{(k)}$ and $d_2^{(k)}>\rad_{min}^{(k)},\forall t\in \R_0^+$. 
 Substituting these inequalities into the radius expression in Equation~\eqref{eqn:rad_sol}, we obtain for all time $t$:
 \begin{align}
     \rad^{(k)}(t)>\frac{-1}{\nu}\ln(\mathsf{e}^{-\nu \rad_{max}^{(k)}}+2\mathsf{e}^{-\nu \rad_{min}^{(k)}})>0.
 \end{align}
 Hence, the tube radius for all agents remains strictly positive at all times. This concludes the proof that the proposed $\Gamma^{(k)}(t)$ in Equation~\eqref{eqn:stt_ball} satisfies the TRAS specification in \eqref{eqn:satras}.
\end{pf}


  

\begin{table*}[t]
\caption{{STT Parameter Tuning Guidelines}}
\label{tab:tuning_guidelines}
\centering
\renewcommand{\arraystretch}{1.15}
{
\begin{tabular}{p{0.1\textwidth} p{0.85\textwidth}}
\hline
\textbf{Parameter} & \textbf{Guideline} \\
\hline
$h_1^{(k)}$ & Regulates convergence of the tube center to the goal. Larger values lead to more aggressive target reaching. Choose $h_1^{(k)} > 1/t_c^{(k)}$. \\
$h_{2,j}^{(k)},\, h_{3,j}^{(k)}$ & Regulate obstacle avoidance. Larger values lead to stronger deviation away from nearby obstacles. \\
$\hat h_2^{(k,l)},\, \hat h_3^{(k,l)}$ & Regulate inter-agent avoidance. Larger values lead to stronger avoidance between neighboring agents. \\
$b$ & Controls the rate at which the SIF decays to zero after $t = t_c^{(k)}$. Smaller values result in faster decay. \\
$\nu$ & Controls the smooth minimum approximation. Larger values bring it closer to the true minimum. \\
$\rad^{(k)}_{\max}$ & Defines the maximum tube radius and sensing range. Choose it based on how early avoidance should become active, while ensuring $\rad^{(k)}_{\max} \leq \min(d_S^{(k)}, d_T^{(k)})$. \\
$\rad^{(k)}_{\min}$ & Defines the minimum allowable tube radius, with $\leq \rad^{(k)}_{\min} \leq \rad^{(k)}_{\max}$. \\
$\soc^{(k)}$ & Defines the social awareness level of each agent. It can be chosen based on task priority or the desired social behavior (e.g., egoistic or altruistic). \\
\hline
\end{tabular}
}
\end{table*}

\begin{remark}
    {Although the theoretical analysis allows any positive values for the gains used in the center dynamics \eqref{eqn:cen_dynamic} and radius dynamics \eqref{eqn:radius_dynamics}, their practical selection depends on factors such as the available actuation limits, and the desired responsiveness of the STT. Therefore, Table~\ref{tab:tuning_guidelines} provides practical guidelines for tuning these gains, as well as for selecting the minimum and maximum STT radii.}
\end{remark}
\begin{lemma}\label{lemma_c}
    For each agent $k \in \Ag $ the tube center $\cen^{(k)}(t)$, the tube radius $\rad^{(k)}(t)$, and their time derivatives $\dot{\cen}^{(k)}(t), \dot{\rad}^{(k)}(t)$ are continuous and bounded for all time $t$.
\end{lemma}

\begin{proof}
For each agent $k\in \Ag $, as the radius dynamics in \eqref{eqn:radius_dynamics} smoothly approximate the $\min$ function, both $\rad^{(k)}(t)$ and $\dot{\rad}^{(k)}(t)$ remain continuous and bounded at all times. 

Moreover, by the second part of Theorem \ref{theorem_tube} each agent's tube center follows $\norm{\cen^{(k)}(t)-o^{(j)}_p}>\rad^{(j)}_o+\rad^{(k)}_{min},\forall j \in [1;n_o]$ and by the third part of the same theorem we also have for each distinct pairs of agents $\{k,l\}\in \Ag $, $\norm{\cen^{(k)}(t)-\cen^{(l)}(t)}>\rad_{min}^{(k)}+\rad_{min}^{(l)},$. This implies that $\cen^{(k)}(t)$ and $\dot{\cen}^{(k)}(t)$, which smoothly depend on target, unsafe sets, and center of tube of other agents are also continuous and bounded for all time.
\end{proof}

\section{Controller Synthesis}
In this section, we derive a closed-form, approximation-free control law to constrain the system output of each agent within its respective STTs. We show the control law derivation for an arbitrary agent $k\in\Ag $, which can be extended in the same fashion to all other agents. The proposed design is inspired by the prescribed performance control (PPC) framework for strict-feedback systems \cite{PPCfeedback}, which exploits the lower triangular structure of the system dynamics in \eqref{eqn:sysdyn}. In standard PPC, as well as in prior STT-based works, the performance bounds are typically imposed independently along each state dimension. In contrast, the proposed approach enforces a ball-shaped spatiotemporal tube, which introduces a coupled geometric constraint on the output and leads to a different first-stage reference design. In particular, we first construct the intermediate control input $r_2^{(k)}$ to ensure that the system output remains within the STT. Subsequently, following the recursive funnel-based design in \cite{das2024spatiotemporal}, we define intermediate signals $r_{z+1}^{(k)}$ such that each state $x_z$ tracks $r_z$ for all $z\in[2;N]$, with $u^{(k)}=r_{N+1}^{(k)}$ as the final control input for the $k^{th}$ agentThe steps of the control design are as follows:

\textbf{Stage $1$}: Given the STT for agent $k$ $\Gamma^{(k)}(t)$, define the normalized and transformed errors $e^{(k)}_1(x_1,t)$ and $\varepsilon^{(k)}_1(x_1,t)$,
\begin{align}
e^{(k)}_1(x^{(k)}_1,t)& = \frac{\norm{x^{(k)}_1(t) - \cen^{(k)}(t)}}{\rad^{(k)}(t)}\notag\\
\varepsilon_1^{(k)}(x^{(k)}_1,t)& = \ln\left( \frac{1 + e^{(k)}_1(x_1^{(k)},t)}{1 - e^{(k)}_1(x_1^{(k)},t)} \right).\nonumber
\end{align}
The intermediate control input $r_2(x_1,t)$ is then given by 
\begin{align}
    r^{(k)}_2(x_1,t) = -\kappa^{(k)}_1 \varepsilon^{(k)}_1(x^{(k)}_1,t) \left( x^{(k)}_1(t)-\cen^{(k)}(t) \right),  
    \kappa_1^{(k)} \in \R^+.
\end{align}

\textbf{Stage $z$} ($z \in [2;N]$): 
To ensure $x_z$ tracks the reference signal $r_z$ from Stage $z-1$, we define a time-varying bound: $\gamma^{(k)}_{z,i}(t) = (p^{(k)}_{z,i} - q^{(k)}_{z,i})e^{-\mu^{(k)}_{z,i}t} + q^{(k)}_{z,i}$, and enforce,
    $-\gamma^{(k)}_{z,i}(t) \leq (x^{(k)}_{z,i}-r^{(k)}_{z,i}) \leq \gamma^{(k)}_{z,i}(t), \forall (t,i) \in \R_0^+ \times [1;n],$
where, $\mu^{(k)}_{z,i} \in \R_0^+$, and $p^{(k)}_{z,i}> q^{(k)}_{z,i} \in \R^+$ are chosen such that $|x^{(k)}_{z,i}(0) - r^{(k)}_{z,i}(0)| \leq p^{(k)}_{z,i}$.
Now, define the normalized error and transformed errors $\varepsilon_z^{(k)}(x^{(k)}_{z},t)$ and $\xi^{(k)}_z(x^{(k)}_{z},t)$ as
\begin{subequations} \label{eq:stage k}
    \begin{align}
    e^{(k)}_z(x^{(k)}_{z},t) &= [e^{(k)}_{z,1}(x^{(k)}_{z,1},t), \ldots, e^{(k)}_{z,n}(x^{(k)}_{z,n},t)]^\top 
    = (\textsf{diag}(\gamma^{(k)}_{z,1}(t),\ldots,\gamma^{(k)}_{z,n}(t)))^{-1} \left(x^{(k)}_{z} - r^{(k)}_z \right),  \\
    \varepsilon^{(k)}_z(x^{(k)}_{z},t) &= \big[\ln\left(\frac{1+e^{(k)}_{z,1}(x^{(k)}_{z,1},t)}{1-e^{(k)}_{z,1}(x^{(k)}_{z,1},t)}\right), \ldots,\ln\left(\frac{1+e^{(k)}_{z,n}(x^{(k)}_{z,n},t)}{1-e^{(k)}_{z,n}(x^{(k)}_{z,n},t)}\right) \big]^\top, \\
    \xi^{(k)}_z(x^{(k)}_{z},t) &=  4(\textsf{diag}(\gamma^{(k)}_{z,1}(t),\ldots,\gamma^{(k)}_{z,n}(t)))^{-1}\Big(I_n-\textsf{diag}(e_z^{(k)}\circ e_z^{(k)})\Big)^{-1} .
\end{align}
\end{subequations}
The next intermediate control input $r^{(k)}_{z+1}(\overline{x}^{(k)}_{z},t)$ is then:
\begin{equation*}
    r^{(k)}_{z+1}(\overline{x}_{k},t) = - \kappa^{(k)}_z\xi^{(k)}_z(x^{(k)}_{z},t)\varepsilon^{(k)}_z(x^{(k)}_{z},t), \kappa^{(k)}_z \in \R^+.
\end{equation*}
At stage $N$, this intermediate input is the actual control input:
\begin{equation*}
    u^{(k)}(\overline{x}^{(k)}_{N},t) = - \kappa^{(k)}_N\xi^{(k)}_N(x^{(k)}_{N},t)\varepsilon^{(k)}_N(x^{(k)}_{N},t), \kappa^{(k)}_N \in \R^+.
\end{equation*}

We now state the main theorem guaranteeing that this controller enforces the desired TRAS behavior.

\begin{theorem} \label{theorem_ras}
    Consider the SA-MAS in \eqref{eqn:social_mas}, where each system $(\Sigma^{(k)},s^{(k)}_a$ in \eqref{eqn:sysdyn} has an associated social awareness index $\soc^{(k)}\in(0,1)$ satisfying Assumptions \ref{assum:lip} and \ref{assum:pd}, a temporal reach-avoid-stay (TRAS) specification as defined in Definition~\ref{def:satras}, and the corresponding STT of the agent $\Gamma^{(k)}(t)$ as derived in Equation~\eqref{eqn:stt_ball} while considering social awareness indexes.
    
    If the initial output is within the STT at time $t=0$: $y^{(k)}(0) \in \Gamma^{(k)}(0)$, 
    then the closed-form control laws
    \begin{align}\label{eqn:Control_ras}
        r^{(k)}_2(x^{(k)}_1,t) &= -\kappa^{(k)}_1 \varepsilon^{(k)}_1(x^{(k)}_1,t) \left( x^{(k)}_1(t)-\cen^{(k)}(t) \right), \notag \\
        r^{(k)}_{z+1}(\bar x^{(k)}_{z},t) &= - \kappa^{(k)}_z\xi^{(k)}_z(x^{(k)}_{z},t)\varepsilon^{(k)}_z(x^{(k)}_{z},t), z \in [2;N-1] \notag\\
        u^{(k)}(\bar x_{N},t) &= - \kappa^{(k)}_N\xi^{(k)}_N(x^{(k)}_{N},t)\varepsilon^{(k)}_N(x^{(k)}_{N},t),
    \end{align}    
    where $\kappa^{(k)}_1, \kappa^{(k)}_z , \kappa^{(k)}_N  \in \R^+$ are the control gains and the control input scheme in \eqref{eqn:Control_ras} ensure that the system output remains within the STT: 
    $$y^{(k)}(t) = x^{(k)}_1(t) \in \Gamma^{(k)}(t), \forall t \in \R_0^+,$$
    thereby satisfying the desired TRAS specification.
\end{theorem}
\begin{proof}
We will prove the correctness of the control law for any agent $k\in \Ag$ and the same reasoning can be extended to other agent. As in \cite{das2024spatiotemporal} we also proof the correctness of control law for Stage 1, and for Stages 2 through N. To simplify notation, we will omit the superscript $(k)$ and treat $\cen^{(k)}(t)$ as $\cen(t)$, $\rad^{(k)}(t)$ as $\rad^{(k)}(t)$ and we also drop the superscript $k$ from the system dynamics \eqref{eqn:sysdyn_ind}. Since the analysis in this section applies uniformly to all agents, this change will not impact readability.

\textbf{Stage $1$:}  
Differentiating $e_1(x_1,t)$ with respect to time $t$ and substituting the system dynamics from \eqref{eqn:sysdyn}, we obtain:
\begin{align}
    \dot{e}_1(x_1,t) &= \Big( \|x_1 -\cen\|^{-1}(x_1 - \cen)^\top(f_1(x_1) + g_1(x_1)x_2 - \dot{\cen}(t)) - \dot{\rad}(t)e_1(x_1,t) \Big) / \rad^{(t)} := h_1(e_1,t).
\end{align}
We define the error constraints for $e_1$ through the open and bounded set $\mathbb{D}:=(0,1)$.
The proof proceeds in three steps. First, we show that a maximal solution exists within $\mathbb{D}$ in the maximal time solution interval $[0, \tau_{\max})$. Next, we prove that the proposed control law \eqref{eqn:Control_ras} ensures $e_1(x_1,t)$ remains in a compact subset of $\mathbb{D}$. Finally, we prove that $\tau_{\max}$ can be extended to $\infty$.

\underline{\textit{Step (i):}}  
Given $\|x_1(0) - \cen(0)\| \leq \rad(0)$, the initial error $e_1(x_1(0),0)$ lies in $\mathbb{D}$. Since $\cen(t)$, $\rad(t)$ are smooth and bounded (Lemma~\ref{lemma_c}), $f_1(x_1)$, $g_1(x_1)$ are locally Lipschitz, and the control law $r_2(x_1,t)$ is smooth in $\mathbb{D}$, the dynamics $h_1(e_1,t)$ is locally Lipschitz in $e_1$ and continuous in $t$. Hence, by \cite[Theorem 54]{sontag}, there exists a maximal solution $e_1 : [0, \tau_{\max}) \rightarrow \mathbb{D}$ such that $e_1(t) \in \mathbb{D}$ for all $t \in [0, \tau_{\max})$.

\underline{\textit{Step (ii):}}  
Consider the Lyapunov candidate $V_1 = 0.5\varepsilon_1^{2}$. Differentiating $V_1$ w.r.t. $t$, and substituting $\dot{\varepsilon}_1$, $\dot{e}_1$, and the system dynamics with the control law \eqref{eqn:Control_ras}, we get:
\begin{align*}
    &\dot{V}_1 = \varepsilon_1 \dot{\varepsilon}_1
    =\frac{2\varepsilon_1}{\rad(1-e_1^2)} \Big( \frac{(x_1-\cen)^\top}{\norm{x_1-\cen}}(\dot{x}_1-\dot{\cen})-\dot{\rad}e_1 \Big) \\
    &=  \frac{2\varepsilon_1}{\rad(1-e_1^2)} \Big( \frac{(x_1-\cen)^\top}{\norm{x_1-\cen}}g_1(x_1)x_2+ {\Phi_1}\Big)\\
    &=\frac{2}{\rad(1-e_1^2)}\Big(\frac{-\varepsilon_1 \kappa}{\norm{x_1-\cen}}(x_1-\cen)^\top g_1(x_1)(x_1-\cen)+\varepsilon_1 \Phi_1 \Big),  
\end{align*}
where $ \Phi_1 := \frac{(x_1 - \cen)^\top}{\|x_1 - \cen\|}(f_1(x_1) + d_1 - \dot{\cen}(t) - \dot{\rad}(t)e_1).$
Using the Rayleigh-Ritz inequality and Assumption~\ref{assum:pd}, we have,
$\underline{g}_1\norm{x_1-\cen}^2 \leq \lambda_{min}(g_1(x_1))\norm{x_1-\cen^2}^2 \leq (x_1-\cen)^\top g_1(x_1)(x_1-\cen)$,
which leads to:
\begin{align*}
\dot{V}_1 &\leq \alpha_1 \left(-\kappa \varepsilon_1^2 \underline{g}_1 \|x_1 - \cen\|^2 + \varepsilon_1 \|\Phi_1\| \right), 
\end{align*}
where $\alpha_1 = \frac{2}{\rad(1 - e_1^2)} > 0$.
From Lemma~\ref{lemma_c}, the functions $\cen(t)$, $\dot{\cen}(t)$, $\rad(t)$, $\dot{\rad}(t)$ are all bounded. Since $x_1(t)$ remains in the tube by Step (i), and $f_1$, $g_1$ are continuous, it follows that $\|\Phi_1\| < \infty$ for all $t \in [0, \tau_{\max})$.
Now, for some $\theta \in (0,1)$, we add and subtract $\kappa \alpha_1 \varepsilon_1^2 \underline{g}_1 \theta \|x_1 - c\|^2$:
\begin{align*}
       \dot{V}_1& \leq \alpha_1 \Big (-\kappa_1 \varepsilon_1^2 \underline{g}_1 (1-\theta)\norm{x_1-\cen}^2 -(\kappa_1 \underline{g}_1 \varepsilon_1^2 \theta \norm{x_1-\cen}^2- \|\varepsilon_1\| \norm{\Phi_1}) \Big)\nonumber\\
       &\leq -\alpha_1 \varepsilon_1^2 \underline{g}_1 (1-\theta) \norm{x_1-\cen}^2,  \forall \kappa_1 \underline{g}_1 \|\varepsilon_1\| \theta \norm{x_1-\cen}^2-\norm{\Phi_1} \geq 0\nonumber \\
       &\leq  -\alpha_1 \varepsilon_1^2 \underline{g}_1 (1-\theta) \norm{x_1-\cen}^2, \forall \norm{\varepsilon_1}\geq \frac{\norm{\Phi_1}}{\kappa_1 \underline{g}_1 \theta \norm{x_1-\cen}^2} := \varepsilon_1^*, \ \forall t\in [0.\tau_{max}).
\end{align*}

Thus, we can conclude that there exists a time-independent upper bound $\varepsilon_1^* \in \mathbb{R}^+_0$ to the transformed error, i.e., $\|\varepsilon_1\| \leq \varepsilon_1^*$, for all $t\in[0,\tau_{max})$. Inverting the transformation, we can express the bounds on $e_1$ as:
\begin{align*}
    0 \leq e_{1} \leq \overline{e}_{1} := \frac{e_{1}^{\varepsilon_{1}^*}-1}{e_{1}^{\varepsilon_{1}^*}+1} < 1.
\end{align*}
Thus, $e_1(t) \in [0, \overline{e}_{1}] =: \mathbb{D}' \subset \mathbb{D}$ for all $t \in [0, \tau_{\max})$.

\underline{\textit{Step (iii):}}  
Since $e_1(t)$ remains in the compact subset $\mathbb{D}' \subset \mathbb{D}$ for all $t \in [0, \tau_{\max})$, the solution cannot escape $\mathbb{D}$ in finite time. By contradiction, if $\tau_{\max} < \infty$, then by \cite[Prop. C.3.6]{sontag}, there must exist $t' < \tau_{\max}$ such that $e_1(t') \notin \mathbb{D}$, which contradicts Step (ii). Hence, the solution exists for all $t \geq 0$, i.e., $\tau_{\max} = \infty$.

\textbf{Stages $z \in [2, N]$:}  
For the remaining stages, we apply the same reasoning as presented in Theorem 4.1 of \cite{das2024spatiotemporal}, and is thus omitted here for brevity.

This completes the proof, showing that the control law in \eqref{eqn:Control_ras} enforces the tube constraint \eqref{eqn:stt_con}, and thereby ensures satisfaction of the T-RAS task for all agent $k\in \Ag$ respecting the \textit{social index value} of each agent.
\end{proof}
\section{Case Studies}
To validate the effectiveness of the proposed real-time multi-agent STT framework, we present two case studies: a 2D omnidirectional mobile robot and a 3D UAV. In addition to simulation results, we also demonstrate real-world applicability through hardware experiments for the mobile robot case. {The code used for the simulations and comparative studies in this paper is publicly available at \href{https://github.com/FocasLab/STT-for-SA-MAS.git}{https://github.com/FocasLab/STT-for-SA-MAS.git}.}
\subsection{2D Omnidirectional Robot}

\subsubsection{Hardware Experiments}
To demonstrate the real-world applicability of the proposed framework, we conducted hardware experiments using two physical omnidirectional robots, with dynamics adopted from \cite{das2024prescribed}, operating in a cluttered dynamic indoor environment (Figure~\ref{fig:hardware}). The experiments are performed in a real arena as shown in Figure~\ref{fig:hardware_setup} equipped with a motion capture system and a projector-based setup that generates random dynamic obstacles in real time. Each robot receives only local information about nearby obstacles and neighboring agents within its sensing range, and the proposed STT synthesis and control law are implemented onboard on a Raspberry Pi 4 in real time.

\begin{figure}
    \centering
    \includegraphics[width=0.3\linewidth]{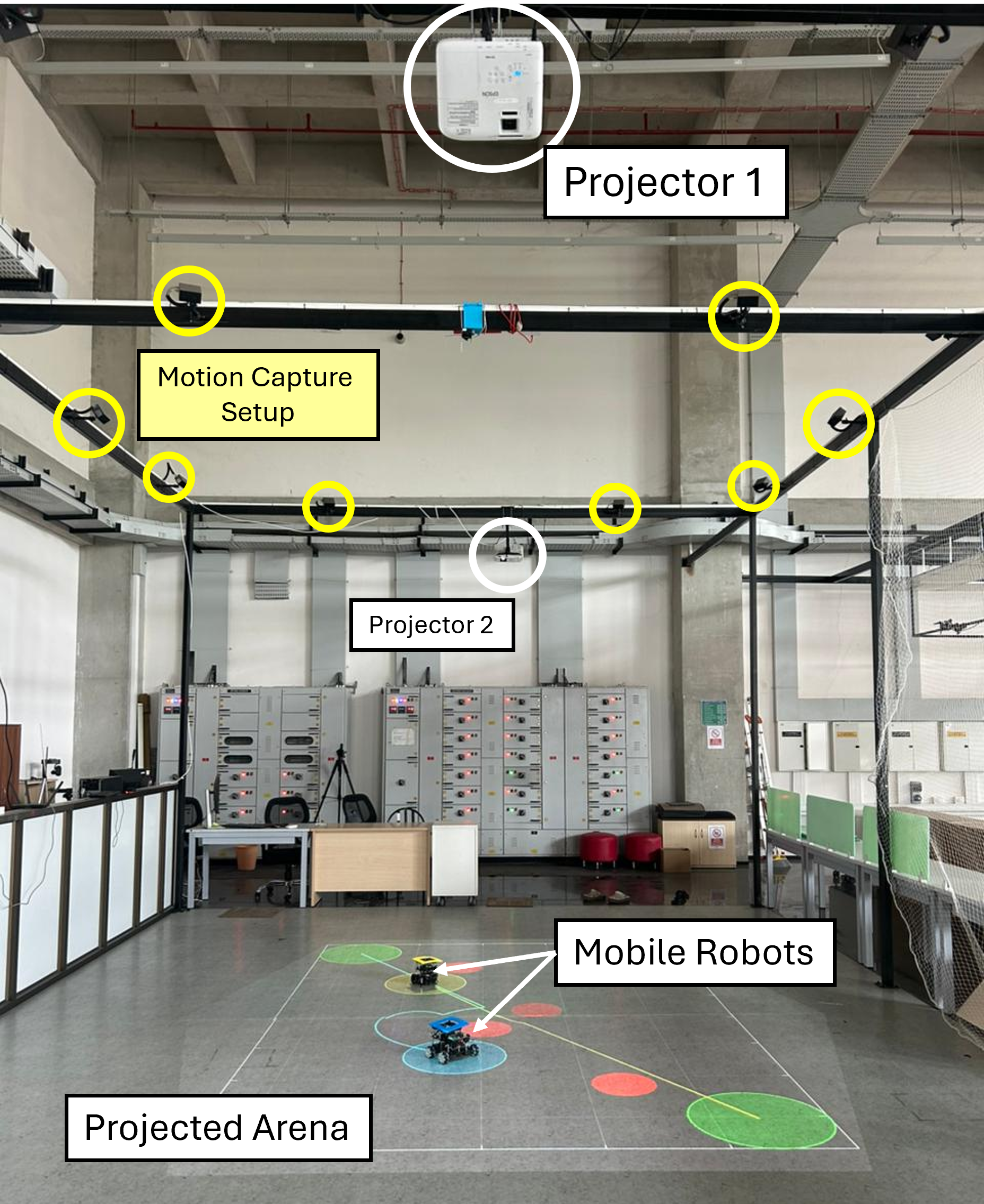}
    \caption{Setup for hardware experiments}
    \label{fig:hardware_setup}
\end{figure}
The task involved position-swapping, where each robot was required to exchange its position with the other within $t_c^{(k)}=120s,\,k\in[1;2]$, while avoiding dynamic obstacles. The tube radii were bounded as $\rad_{min}^{(k)}=0.21 \leq \rad^{(k)} \leq 0.27=\rad_{max}^{(k)}$, and the center dynamics parameters in \eqref{eqn:cen_dynamic} were chosen as $h_1^{(k)}=0.05,\, \hat h_2^{(k,l)}=\hat h_3^{(k,l)}=8,\, k \neq l$, and $h_{2,j}^{(k)}=h_{3,j}^{(k)}=1,\, j \in\mathcal{N}^{(k)}_o(t)$.

To investigate the influence of the social awareness index, we considered two scenarios. In Case~1 (Figure~\ref{fig:subfig1}), the blue agent was assigned a higher social awareness index $\soc^{(\text{blue})}=0.7$ than the yellow agent $\soc^{(\text{yellow})}=0.3$. Consequently, the blue agent behaved more altruistically, showing greater flexibility and cooperation, whereas the yellow agent was more egoistic, prioritizing its goal and obstacle avoidance. In Case~2 (Figure~\ref{fig:subfig2}), the social awareness indices were swapped while keeping all other parameters identical, leading to the opposite behaviors compared to Case~1. A full video of the hardware experiments is available at \href{https://www.youtube.com/watch?v=oDo6Qs9vw7s}{Link}\footnotemark{}.
\subsubsection{Simulation Studies}
To demonstrate the scalability of the proposed framework, we consider a multi-agent system of $n_a = 8$ agents, $\Sigma^{(k)},\, k \in [1;8]$, modeled as omnidirectional mobile robots operating in a 2D environment (Figure~\ref{fig:2d}) with dynamics adapted from~\cite{NAHS}. Each agent is tasked with swapping its position with the diagonally opposite agent. The prescribed time for each agent is different and is defined as $t_c^{(k)} = 5+(k-1)$s, for all $k \in [1;8]$. Agents with lower social awareness index $\soc^{(1)} = \soc^{(5)} = 0.1$ are shown in yellow, while the remaining agents with higher social awareness index $\soc^{(k)} = 0.99,\, k \in [1;8]\setminus\{1,5\}$ are shown in blue.  

The tube radii are bounded by $\rad_{min}^{(k)} = 0.6$ and $\rad_{max}^{(k)} = 0.9$ for all agents. The center dynamics parameters in \eqref{eqn:cen_dynamic} are set as $h_1^{(1)}=0.22,\, h_1^{(2)}=0.15,\, h_1^{(3)}=h_1^{(4)}=0.13,\, h_1^{(k)}=0.12,\, k \in [5,8]$, and $\hat h_2^{(k,l)} = \hat h_3^{(k,l)} = 0.04,\, k \neq l$.  

Figure~\ref{fig:2d} shows the trajectories and STTs of all agents at different time steps. All agents remain safely within their respective tubes and successfully reach their targets. A complete simulation video is available at \href{https://www.youtube.com/watch?v=oDo6Qs9vw7s}{Link}\footnotemark[\value{footnote}].

\begin{figure*}
    \centering
    \begin{subfigure}{\linewidth}
        \centering \includegraphics[width=0.9\linewidth]{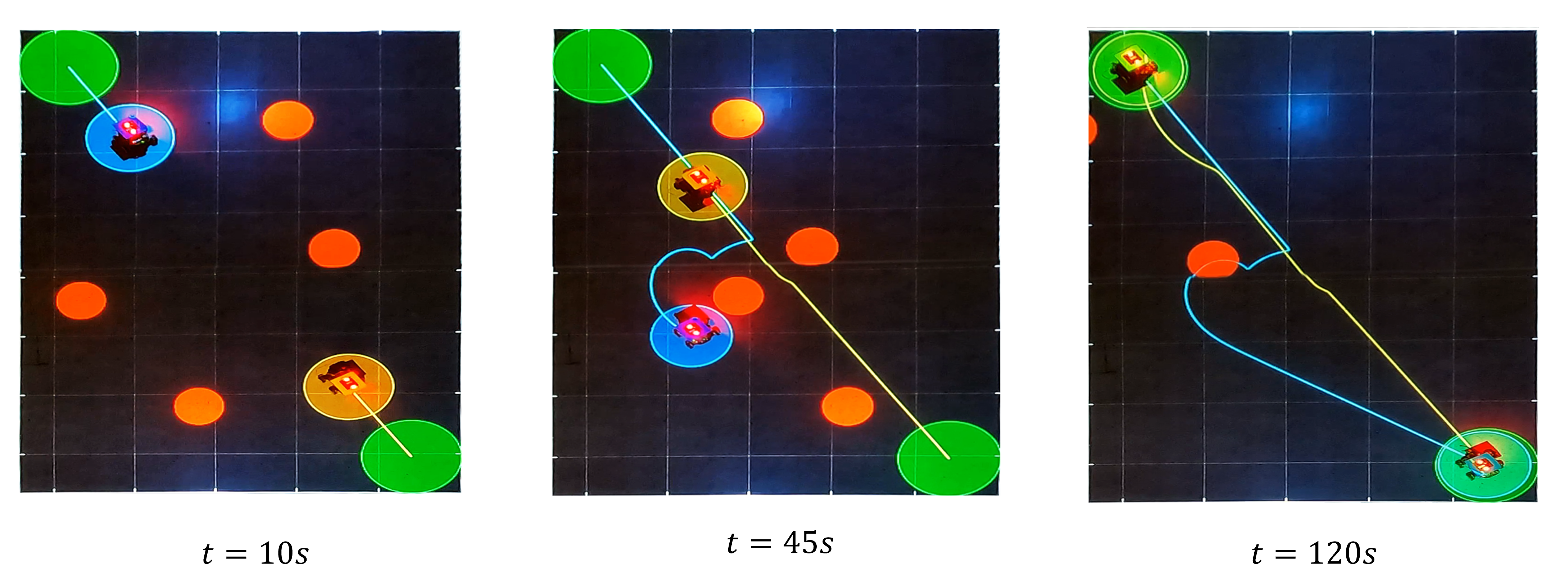}
        \caption{The blue robot is altruistic  $\soc^{(\text{blue})}=0.7$ and the yellow robot is egoistic  $\soc^{(\text{yellow})}=0.3$}
        \label{fig:subfig1}
    \end{subfigure}
    
    \vspace{1em} 

    \begin{subfigure}{\linewidth}
        \centering        \includegraphics[width=0.9\linewidth]{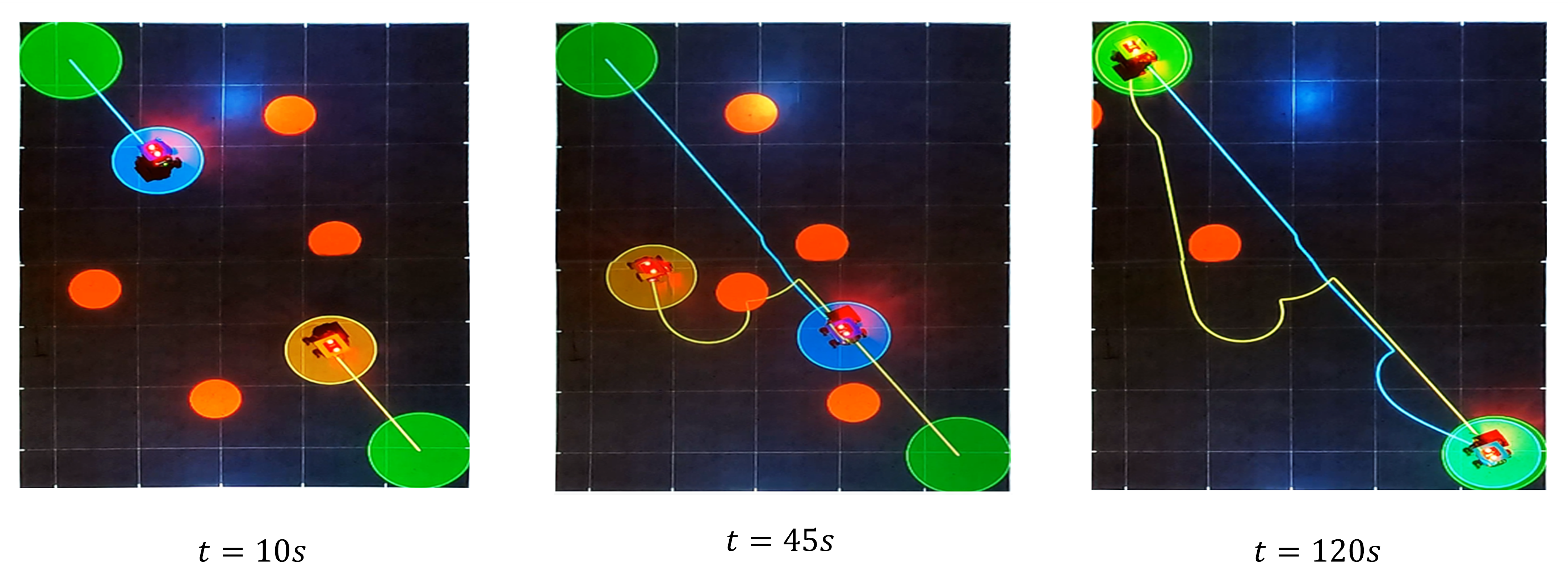}
        \caption{The blue robot is egoistic  $\soc^{(\text{blue})}=0.3$, and the yellow robot is altruistic  $\soc^{(\text{yellow})}=0.7$}
        \label{fig:subfig2}
    \end{subfigure}

    \caption{Hardware demonstration of two omnidirectional robots in a cluttered dynamic environment,\href{https://www.youtube.com/watch?v=oDo6Qs9vw7s}{Video}.}
    \label{fig:hardware}
\end{figure*}

\begin{figure*}
    \centering
    \includegraphics[width=0.9\linewidth]{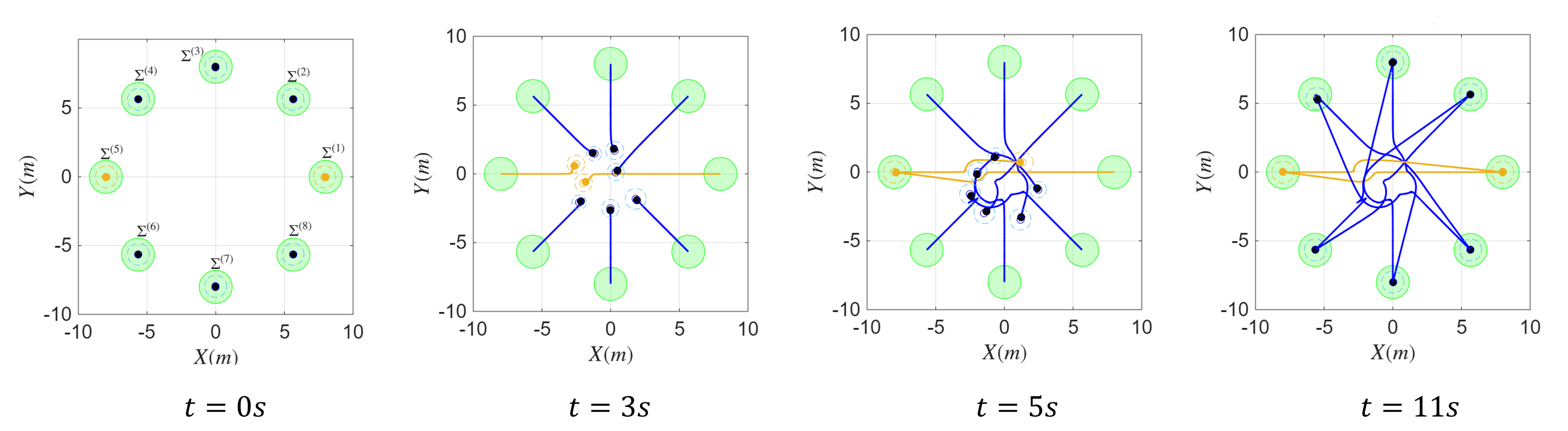}
    \caption{Simulation of eight omnidirectional mobile robots in a 2D environment with different prescribed times. Two egoistic agents ($\soc^{(1)} = \soc^{(5)} = 0.1$, yellow) interact with six altruistic agents ($\soc^{(2)} = \soc^{(3)} = \soc^{(4)} = \soc^{(6)} = \soc^{(7)} = \soc^{(8)} = 0.99$), demonstrating scalability to multiple agents, \href{https://www.youtube.com/watch?v=oDo6Qs9vw7s}{Video}.}
    \label{fig:2d}
\end{figure*}

\subsection{3D Unmanned Aerial Vehicle}
We consider a multi-agent UAV system with $n_a=8$ agents operating in a 3D environment, where each agent follows second-order dynamics adapted from~\cite{APF_drone}.

Each UAV starts from its assigned initial position and is tasked with swapping positions with the diagonally opposite agent. The prescribed time for agent $\Sigma^{(1)}$ is $t_c^{(1)} = 20\,\mathrm{s}$, while all other agents have $t_c^{(k)} = 25\,\mathrm{s}$. Agents with a lower social awareness index $\soc^{(1)} = \soc^{(5)} = 0.1$ are shown in yellow, whereas agents with a higher social awareness index $\soc^{(k)} = 0.99,\, k \in [1;8]\setminus\{1,5\}$ are shown in blue.

The tube radii are bounded by $\rad_{min}^{(k)} = 0.6$ and $\rad_{max}^{(k)} = 0.9$ for all agents. The center dynamics parameters in \eqref{eqn:cen_dynamic} are set as $h_1^{(1)} = 0.07,\, h_1^{(k)} = 0.06,\, k \in [2,8]$ and $\hat h_2^{(k,l)} = \hat h_3^{(k,l)} = 0.002,\, k \neq l$.

Figure~\ref{fig:3d} shows the trajectories and STTs of all agents at various time steps. All UAVs remain strictly within their tubes, successfully reach their targets, and avoid inter-agent collisions according to their social indices. A full simulation video is available at \href{https://www.youtube.com/watch?v=oDo6Qs9vw7s}{Link}\footnotemark[\value{footnote}].
\begin{figure*}
    \centering
    \includegraphics[width=0.9\linewidth]{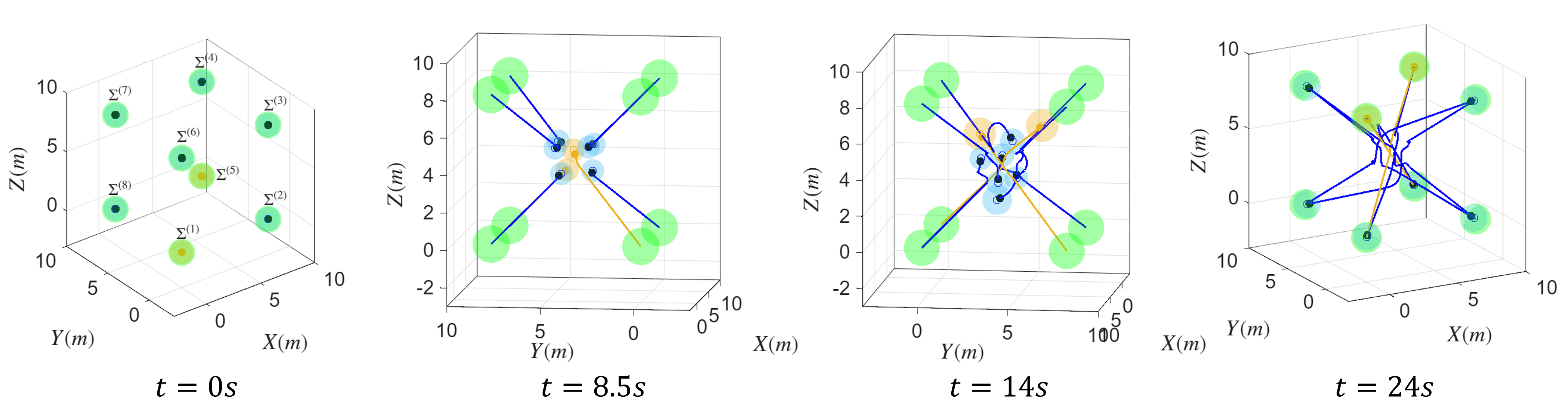}
    \caption{Simulation of eight UAVs in a 3D environment with different prescribed times. Two egoistic agents ($\soc^{(1)} = \soc^{(5)} = 0.1$, yellow) and six altruistic agents ($\soc^{(2)} = \soc^{(3)} = \soc^{(4)} = \soc^{(6)} = \soc^{(7)} = \soc^{(8)} = 0.99$) coordinate safely, showing scalability from ground robots to aerial systems, \href{https://www.youtube.com/watch?v=oDo6Qs9vw7s}{Video}.}
    \label{fig:3d}
\end{figure*}

\footnotetext{https://www.youtube.com/watch?v=oDo6Qs9vw7s}

\begin{figure*}
    \centering
    \includegraphics[width=\linewidth]{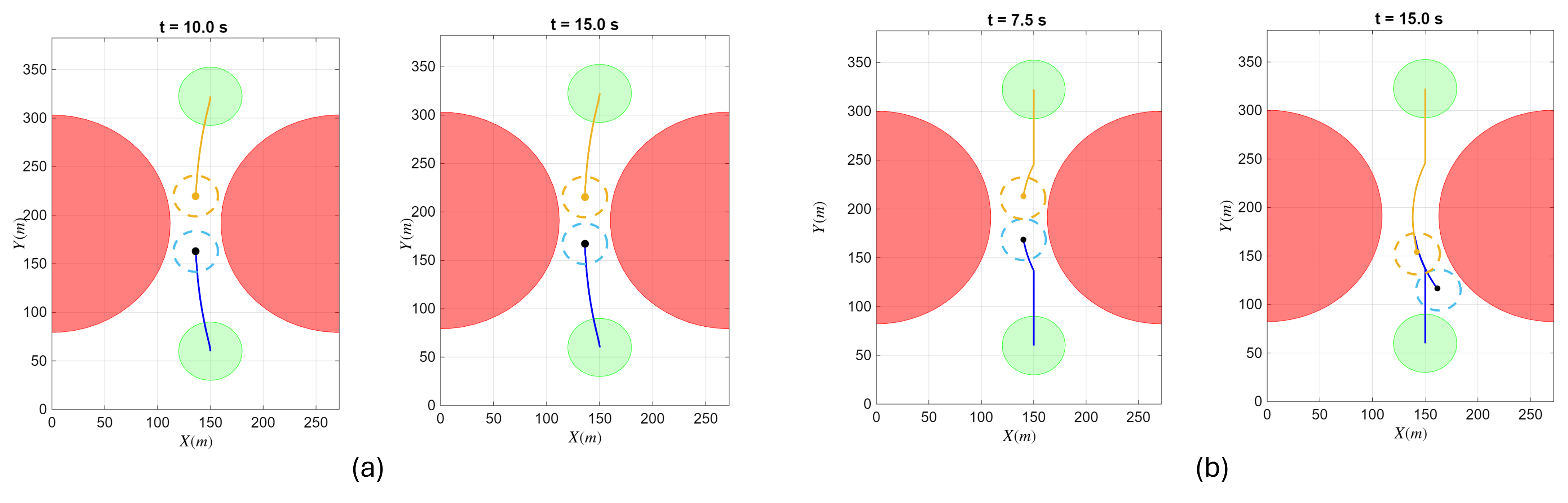}
    \caption{{Comparison of the performance between (a) Safety Barrier Certificates \cite{multi_barrier} and (b) proposed method in a narrow passage scenario, where only one robot can pass at a time.} \href{https://www.youtube.com/watch?v=oDo6Qs9vw7s}{Video}.}
    \label{fig:passage_comp}
\end{figure*}

\begin{table*}[t]
\centering
\begin{adjustbox}{max width=\textwidth}
\begin{threeparttable}
\caption{Comparing Spatiotemporal tubes with classical algorithms}
\begin{tabular}{lcccccccccccc}
\hline
\textbf{Algorithm} & \multicolumn{2}{c}{\textbf{\begin{tabular}[c]{@{}l@{}}Closed-form \\ Control \end{tabular}}} & \multicolumn{2}{c}{\textbf{\begin{tabular}[c]{@{}l@{}}Formal \\ Guarantee \end{tabular}}} & \multicolumn{2}{c}{\textbf{\begin{tabular}[c]{@{}l@{}}Unknown \\ Dynamics \end{tabular}}} & \multicolumn{2}{c}{\textbf{\begin{tabular}[c]{@{}l@{}}Social\\ Awareness \end{tabular}}} & \multicolumn{2}{c}{\textbf{\begin{tabular}[c]{@{}l@{}}Bounded \\ Disturbance \end{tabular}}} & \multicolumn{2}{c}{\textbf{\begin{tabular}[c]{@{}l@{}}Time  \\ Constraint \end{tabular}}} \\ \hline
Safety Barrier Certificates\cite{multi_barrier} & \multicolumn{2}{c}{\xmark} & \multicolumn{2}{c}{\cmark} & \multicolumn{2}{c}{\xmark} & \multicolumn{2}{c}{\xmark} & \multicolumn{2}{c}{\xmark} & \multicolumn{2}{c}{\xmark} \\
Distributed MPC\cite{Multi_MPC} \cite{dai2017distributed} & \multicolumn{2}{c}{\xmark} & \multicolumn{2}{c}{\xmark} & \multicolumn{2}{c}{\xmark} & \multicolumn{2}{c}{\xmark} & \multicolumn{2}{c}{\xmark} & \multicolumn{2}{c}{\xmark} \\ 
Negotiation Based STT \cite{multi_stt} & \multicolumn{2}{c}{\cmark} & \multicolumn{2}{c}{\cmark} & \multicolumn{2}{c}{\cmark} & \multicolumn{2}{c}{\xmark} & \multicolumn{2}{c}{\cmark} & \multicolumn{2}{c}{\cmark} \\
RA-CBF \cite{risk_aware} & \multicolumn{2}{c}{\xmark} & \multicolumn{2}{c}{\cmark} & \multicolumn{2}{c}{\xmark} & \multicolumn{2}{c}{\cmark} & \multicolumn{2}{c}{\xmark} & \multicolumn{2}{c}{\xmark} \\
WBVC \cite{pierson2020weighted} & \multicolumn{2}{c}{\xmark} & \multicolumn{2}{c}{\xmark} & \multicolumn{2}{c}{\xmark} & \multicolumn{2}{c}{\cmark} & \multicolumn{2}{c}{\xmark} & \multicolumn{2}{c}{\xmark} \\
Social MAPF \cite{chandra2023socialmapf} & \multicolumn{2}{c}{-\tnote{1}} & \multicolumn{2}{c}{\cmark} & \multicolumn{2}{c}{\xmark} & \multicolumn{2}{c}{\cmark} & \multicolumn{2}{c}{\xmark} & \multicolumn{2}{c}{\xmark} \\
SAMARL \cite{SAMRL} & \multicolumn{2}{c}{\xmark} & \multicolumn{2}{c}{\xmark} & \multicolumn{2}{c}{\cmark} & \multicolumn{2}{c}{\cmark} & \multicolumn{2}{c}{\cmark} & \multicolumn{2}{c}{\xmark}\\
\textbf{SA-MAS (proposed)} & \multicolumn{2}{c}{\cmark} & \multicolumn{2}{c}{\cmark} & \multicolumn{2}{c}{\cmark} & \multicolumn{2}{c}{\cmark} & \multicolumn{2}{c}{\cmark} & \multicolumn{2}{c}{\cmark} \\
\hline
\end{tabular}
\label{tab:comp}
\begin{tablenotes}
    \item [1] Additional mechanisms like PID and MPC are required for control.
  
\end{tablenotes}
\end{threeparttable}
\end{adjustbox}
\end{table*}

\subsection{{Comparative Study}}

{To evaluate the effectiveness of the proposed STT-based framework for Socially-Aware Multi-Agent Systems (SA-MAS), we perform both qualitative and quantitative analyses against existing state-of-the-art methods.
\\
\subsubsection{Qualitative Comparison}
We first provide a qualitative comparison of the proposed framework with existing multi-agent methods in Table~\ref{tab:comp}, including both socially-aware approaches and methods that do not explicitly account for social awareness. In particular, the table compares the proposed method with Safety Barrier Certificates \cite{multi_barrier}, Decentralized MPC, negotiation-based STT methods, RA-CBF \cite{risk_aware}, WBVC \cite{pierson2020weighted}, Social MAPF \cite{chandra2023socialmapf}, and SAMARL \cite{SAMRL}, highlighting their relative strengths and limitations in terms of computation efficiency, formal guarantees, robustness to unknown dynamics and disturbances, social awareness, and prescribed-time guarantees.
}
\begin{table*}[t]
\caption{{Comparison with baseline algorithms based on average computation time per step (ms).}}
\label{tab:comparison_quant}
\centering
\renewcommand{\arraystretch}{1.2}
\resizebox{\textwidth}{!}{%
\begin{tabular}{lcccccc}
\hline
\multirow{3}{*}{\textbf{Method}} &
\multicolumn{3}{c}{\textbf{2D Mobile Robot}} &
\multicolumn{3}{c}{\textbf{3D Quadrotor}} \\
\cline{2-7}
& $n_a$=5 & $n_a$=10 & $n_a$=20
& $n_a$=5 & $n_a$=10 & $n_a$=20 \\
\hline
RA-CBF \cite{risk_aware}
& $14.7\pm0.7$   & $75.24\pm4.9$    & $352.35\pm15.51$
& $25.79\pm1.102$  & $107.14\pm8.307$     & $476.13\pm64.108$ \\
WBVC \cite{pierson2020weighted}
& $0.09\pm0.009$ & $0.37\pm0.03$    & $2.09\pm0.09$
& $2.13 \pm 0.506$  & $5.89\pm 1.62$   & $13.56\pm 2.18$ \\
\textbf{Proposed}
& $\mathbf{0.06\pm0.0003}$ & $\mathbf{0.18\pm0.002}$    & $\mathbf{0.86\pm0.01}$
& $\mathbf{0.06\pm0.0028}$    & $\mathbf{0.16\pm 0.0137}$     & $\mathbf{0.64\pm0.0497}$ \\
\hline
\end{tabular}%
}
\end{table*}

\begin{figure*}[t]
    \centering
    \includegraphics[width=0.8\linewidth]{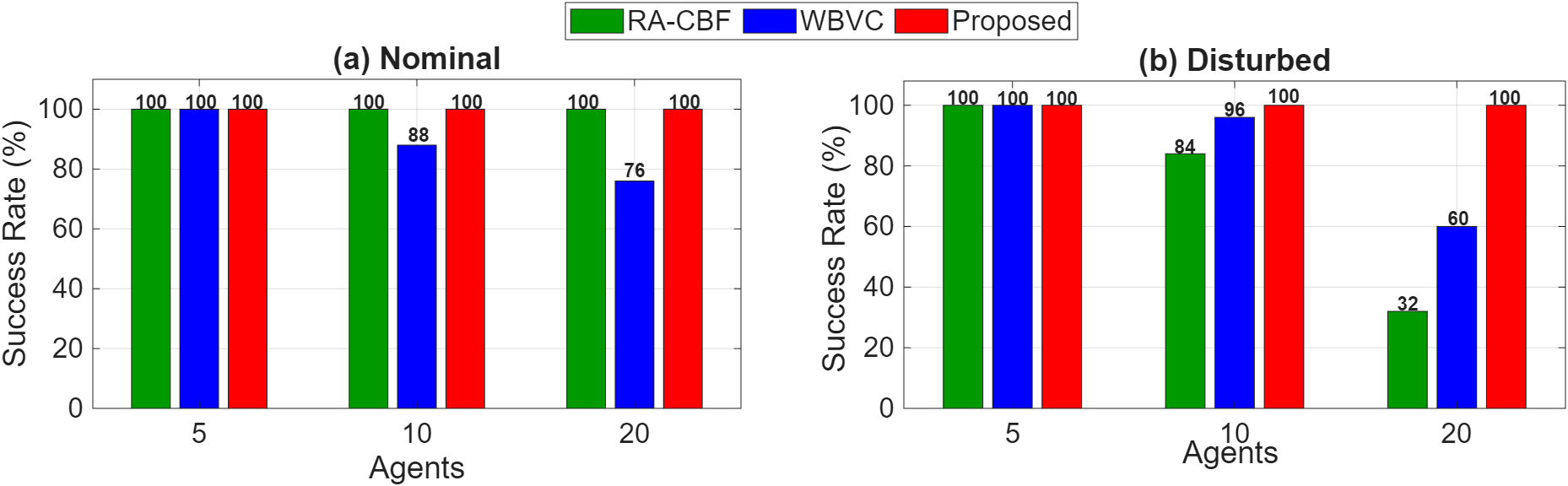}
    \caption{{Comparison of success rate for 2D mobile robot case with 25 randomized runs.}}
    \label{fig:succes_rate}
\end{figure*}

{
\subsubsection{Experimental Scenarios}
The experimental evaluation is divided into two parts to highlight the proposed framework's performance.
\\
\textbf{Part 1: Narrow passage scenario.} 
In the first scenario, we compare our socially-aware framework with Safety Barrier Certificates \cite{multi_barrier}, a representative standard baseline without explicit modeling of social awareness. As shown in Figure~\ref{fig:passage_comp}, two 2D mobile robots must swap positions through a narrow passage that only fits one agent at a time.  
Figure~\ref{fig:passage_comp}(a) shows the trajectories under Safety Barrier Certificates at $t=10\,\text{s}$ and $t=15\,\text{s}$. Both agents approach the passage and eventually get stuck. In contrast, with the proposed method and social awareness values $\soc^{yellow}=0.1$ and $\soc^{blue}=0.9$, the more altruistic blue agent takes a detour, allowing the more egoistic yellow agent to pass first, and then proceeds to its own target, as shown in Figure~\ref{fig:passage_comp}(b). A complete simulation video is available at \href{https://youtu.be/UnNlvu0TdtI}{Link}.}\footnote{https://www.youtube.com/watch?v=UnNlvu0TdtI}

{\textbf{Part 2: Quantitative performance study.}
In the second part of the study, we evaluate our method against socially-aware baselines RA-CBF \cite{risk_aware}, and WBVC \cite{pierson2020weighted}. The evaluation is performed over 25 runs for both 2D mobile robots and 3D quadrotors, with the number of agents chosen as $n_a = 5, 10,$ and $20$. In each run, agents are assigned the task of swapping positions. For the 2D case, initial positions are sampled within a circle of random radius, while for the 3D case, they are sampled within a cube of random side length. The methods are evaluated using two metrics: average computation time per step (in ms) and success rate, defined as the percentage of collision-free runs, under both nominal conditions and bounded disturbances.
The results in Table~\ref{tab:comparison_quant} and Figure~\ref{fig:succes_rate} highlight the following advantages:
\begin{enumerate}
    \item \textbf{Computational Efficiency:} As reported in Table~\ref{tab:comparison_quant}, the proposed method achieves the lowest average computation time per step, mainly due to its closed-form tube synthesis and control design.
    \item \textbf{Scalability:} Table~\ref{tab:comparison_quant} also reports that as the system dimension increases from 2D to 3D, and as the number of agents grows, the proposed framework scales more favorably than the baseline methods.
    \item \textbf{Safety and Robustness:} Figure~\ref{fig:succes_rate} shows that the proposed approach maintains a $100\%$ success rate across all tested scenarios, both under nominal conditions and in the presence of bounded disturbances. In contrast, the performance of RA-CBF and WBVC degrades under disturbances.
    \item \textbf{Prescribed-Time Guarantee:} Unlike the baseline methods, the proposed framework guarantees task completion within the prescribed time, independent of the number of agents or the choice of controller gains.
\end{enumerate}
Direct quantitative comparisons with Social MAPF \cite{chandra2023socialmapf} and SAMARL \cite{SAMRL} are not included, since they are not directly comparable to the setting considered in this work. In particular, Social MAPF is a planning-based approach that operates in discrete state and action spaces and requires an additional tracking controller for execution, while SAMARL is a reinforcement learning–based method that requires extensive offline training and does not provide formal guarantees.
}

\section{Conclusion}
In this work, we propose a real-time spatiotemporal tube-based framework to address the reach-avoid problem for multi-agent systems with unknown dynamics.  Unlike existing approaches that neglect social awareness and assume symmetric interactions, our method accounts for the varying social awareness of individual agents, leading to asymmetric interaction modeling. The framework synthesizes tubes in real time and ensures safe avoidance of unsafe regions during the execution of the assigned task. We provide theoretical guarantees by showing that tubes of each agent never intersect with any other agents and also that they never intersect with any moving obstacles and reach the target within the prescribed time. The resulting control law, which constrains each agent's system output within the STTs, is closed-form and approximation-free. Finally, the effectiveness and scalability of the method were demonstrated through simulations, hardware experiments and comparison studies with multi agent system's of 2D mobile robots and  3D drone navigating through cluttered and dynamic environments.

{While the framework is effective, there are some limitations that will be addressed in future work. First, although in our hardware experiments the control inputs remained within feasible limits, the current framework does not explicitly enforce hard input constraints. In practice, actuator saturation may affect the ability to strictly maintain the "stay inside the tube" guarantee and achieve the desired prescribed-time performance. Extending the framework to explicitly handle input constraints and derive feasibility conditions relating prescribed time, obstacle velocities, and actuator limits is an important direction for future work.
}
Second, we consider predefined social indices capturing the behavior of each agent. Future work may explore optimizing these values using game-theoretic methods to achieve more adaptive and balanced interactions. 
{While in this work, the STT is modeled as a ball in the $n$-dimensional output space, future work will consider extending this representation to asymmetric shapes, such as ellipsoids or polytopes.}
\bibliographystyle{unsrt} 
\bibliography{sources} 

\appendix

\end{document}